\documentclass[11pt]{article}
\usepackage{amssymb}
\usepackage{amsmath}
\usepackage{latexsym}
\usepackage{amsthm}
\usepackage{epic}
\usepackage{graphicx}
  \makeatletter
\def\thmhead@plain#1#2#3{%
 \thmname{#1}\thmnumber{\@ifnotempty{#1}{
 }#2}%
 \thmnote{ \the\thm@notefont(#3)}}
\let\thmhead\thmhead@plain
\def\swappedhead#1#2#3{%
 \thmnumber{#2}\thmname{\@ifnotempty{#2}{. }#1}%
 \thmnote{ \the\thm@notefont(#3)}}
\makeatother

\theoremstyle{definition} 
\newtheorem{definition}{Definition}[section]

\theoremstyle{plain}      
\newtheorem{proposition}[definition]{Proposition}
\newtheorem{theorem}[definition]{Theorem}
\newtheorem{corollary}[definition]{Corollary}
\newtheorem{lemma}[definition]{Lemma}

\begin{document}

\begin{center}
{\Large Hadamard Matrices, Quaternions, and the Pearson Chi-square
Statistic\\}\vspace{9pt}

 {\large Abbas M. Alhakim}\vspace{6pt}

 {\large American University of Beirut}\vspace{6pt}

 {\large Email: aa145@aub.edu.lb}
\end{center}
\begin{abstract}
We present a symbolic decomposition of the Pearson chi-square
statistic with unequal cell probabilities,  by presenting
Hadamard-type matrices whose columns are eigenvectors of the
variance-covariance matrix of the cell counts. All of the
eigenvectors have non-zero values so each component test uses all
cell probabilities in a way that makes it intuitively interpretable.

When all cell probabilities are distinct and unrelated we establish
that such decomposition is only possible when the number of
multinomial cells is a small power of $2$. For higher powers of $2$,
we show, using the theory of orthogonal designs, that the targeted
decomposition is possible when appropriate relations are imposed on
the cell probabilities, the simplest of which is when the
probabilities are equal and the decomposition is reduced to the one
obtained by Hadamard matrices.  Simulations are given to illustrate
the sensitivity of various components to changes in location, scale
skewness and tail probability, as well as to illustrate the
potential improvement in power when the cell probabilities are
changed.
 \end{abstract}

\section{Introduction}
Consider an i.i.d. sequence of random variables
$\xi_0,\xi_1,\ldots,\xi_n$ of trials taking values in the set
$1,\ldots,m$ with respective probabilities $p_1,\ldots,p_{m}$. These
values can be the result of binning continuous random variables into
$m$ multinomial cells. Let $k\geq2$ be an integer and consider the
sequence of dependent random vectors
$\{X_i=(\xi_i,\ldots,\xi_{i+k-1})\}$ for $i=0,\ldots,n$, where the
index $i$ is taken modulo $n$, i.e. the sequence is wrapped around.
Based on this overlapping sequence of dependent and identically
distributed random vectors, a naive chi-square statistic can be
formed as
\[
V_k=\sum\left[C(\alpha)-nP(\alpha)\right]^2/nP(\alpha),
\]
where the sum is taken over all patterns $\alpha$ of size $k$ and
$C(\alpha)$ is the frequency of $\alpha$ in the sequence $X_i$ while
$P(\alpha)$ designates the probability of $\alpha$. $V_k$ mimics the
Pearson chi-square statistic, but due to overlapping, it is
obviously not asymptotically distributed as a chi-square random
variable. I. J. Good \cite{Good53} introduces the difference
statistic $\nabla V_k=V_k-V_{k-1}$ and shows that it has an
asymptotic chi-square distribution with $m^k-m^{k-1}$ degrees of
freedom, when each component in the multinomial probability vector
is equally $1/m$. Indeed, Good \cite{Good53} also shows that
$\nabla^2 V_k=\nabla V_k-\nabla V_{k-1}$ has a chi-square
distribution. Marsaglia proposes $\nabla V_k$ as a test of
uniformity for pseudo-random number sequences capable of detecting
potential local correlations.

The idea of drawing information from the sequence $\{X_i\}$ is also
exploited in the concept of approximate entropy, ApEn, introduced by
Pincus \cite{Pincus91}. ApEn has also been used as a test of
randomness as well as to quantify irregularity in medical time
series in order to predict disorder in patients, see
\cite{Pincus96}.

To thoroughly understand the difference statistic $\nabla V_k$ (and
its relationship to approximate entropy), Alhakim~\cite{Alhakim2004}
analyzes the asymptotic covariance matrix $B_k$ of the simple Markov
chain $\{X_i\}$, obtaining the eigenvalues and a complete set of
orthogonal eigenvectors for the equally likely multinomial
probabilities. This spectral decomposition gives much insight about
the information obtained by increasing the word size from $k-1$ to
$k$. As an example, the kernel of $B_k$ is shown to have dimension
$m^{k-1}$, which explains the degrees of freedom of the difference
statistic. Also, each eigenvector is seen to provide a component
test statistic of $\nabla V_k$ which is asymptotically chi-squared
distributed with exactly one degree of freedom, thus obtaining a
partitioning of $\nabla V_k$ into component tests that are
asymptotically independent. More interestingly, the analysis in
Alhakim~\cite{Alhakim2004} revealed exact relationships between
component tests corresponding to word sizes $k$ and $k-1$. Following
the diagram below with $k=4$, a component test corresponding to an
eigenvector $v_4$ of $B_4$ with eigenvalue $\lambda=4$ is
deterministically the same as an component test with an appropriate
eigenvector $v_3$ of $\lambda=3$ of $B_3$, where $v_4$ can be
obtained recursively from $v_3$. Inductively, this means that all
components for $k=4$ and $\lambda=4$ correspond to components of
$k=1$ and $\lambda=1$. Similarly, all components for $k=4$ and
$\lambda=3$ correspond to components of $k=2$ and $\lambda=1$, etc.
Two important implications are that

(a) The most significant components are those of the simple
non-overlapping case of $k=1$ which provide the components of the
regular Pearson statistic. So it is not useful to use a test with
$k=2$ if the original test with $k=1$ is failed.

(b) The only new information that one obtains when the word size is
increased by one is within the eigenvectors corresponding to the
lowest eigenvalue $1$.

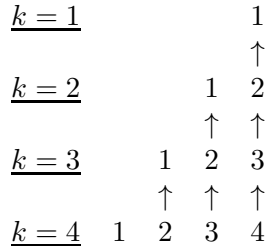
\begin{center}
\begin{tabular}{lllll}
    $\underline{k=1}$&&&&1\\
    &&&&$\uparrow$\\
    $\underline{k=2}$&&&1&2\\
    &&&$\uparrow$&$\uparrow$\\
    $\underline{k=3}$&&1&2&3\\
    &&$\uparrow$&$\uparrow$&$\uparrow$\\
    $\underline{k=4}$&1&2&3&4
\end{tabular}
\end{center}

We add to the previous points the fact that, for any $k$, orthogonal
eigenvectors of $\lambda=1$ can also be obtained recursively from
orthogonal eigenvectors of $\lambda=1$ with word size $k=1$. It is
worth mentioning that Alhakim \textit{et. al} \cite{Alhakim2004b}
generalizes the above results to sequences of random variables
$\xi_i$ that follow any distribution (possibly continuous),
including non-equally likely multinomial probability. Complete
spectral analysis is also obtained with similar relations between
eigenvectors. The above discussion thus suggests that a good set of
orthogonal eigenvectors that partition the regular Pearson statistic
is an essential requirement to get a complete and meaningful
partitioning for higher $k$, both in the case of equally likely and
general non-equally likely multinomial vectors. This is the main
motivation and objective of the current paper.

The paper is organized as follows. We begin by introducing the
classical partition problem in Section~\ref{S:classical}, then we
establish the special case of equally likely cell probabilities
using the non-constant columns of Hadamard matrices in
Section~\ref{S:prelims}. Section~\ref{S:nonequiprobable} deals with
the non-equally likely case. Section~\ref{S:Latin Hadamard}
introduces an orthogonal matrix that serves the purpose of Hadamard
matrix for non-equilikely multinomial probability vectors. In
Section~\ref{S:PowSim} we construct some component vectors and
perform power simulation illustrating the validity and utility of
these components. Section~\ref{S:DivisionAlgebra} explains the
limitations of the method and shows how to overcome these
limitations by adding some restrictions on the cell probabilities.
The last section is an appendix that includes some technical proof.

\section{The Classical Partition Problem}\label{S:classical}
The problem of decomposing the Pearson chi-square statistic into
component tests that are asymptotically independent is a rather old
problem. In fact, Pearson himself addressed this problem in his
Editorial in 1917, see Lancaster \cite{Lancaster1969}, as an answer
to the objection that $X^2$ is an omnibus, overall test. The
importance of the decomposition is that each of the components
provides a test of some particular aspect of the null distribution.
However, the interpretation of these component tests depends largely
on the choice of orthogonal scheme.

Several authors have worked on this problem. Lancaster
\cite{Lancaster1969} presents a general technique of getting various
components based on different orthogonal polynomial schemes. Irwin
\cite{Irwin1949} decomposes $X^2$ into independent components using
Helmert matrices.

More recently, Anderson \cite{Anderson1994} presents an \textit{ad
hoc} decomposition based on $8$-dimensional vectors, whose
orthogonality guarantees the asymptotic independence. He proposes
that his vectors provide component tests that can be used and
interpreted individually as tests for location, symmetry, skewness
and kurtosis. Therefore, when the chi-square test fails, the
individual components provide information about what aspects of the
null distribution leads to rejection. Furthermore, unlike the
columns of a Helmert matrix (and all other decomposition matrices),
the entries in Anderson's vectors all have nonzero values, resulting
in component tests that each use count statistics of all multinomial
cells; a property that should improve the power of these components.

Boero \textit{et al} \cite{Boero2004} notices that these vectors are
obtained from the nonconstant rows of a Hadamard matrix. They then
show how to use a Hadamard matrix, when the number of multinomial
classes, $k$, is a power of $2$, to obtain a full set of $k-1$
asymptotically independent components that include those of
Anderson. Indeed, Boero et al \cite{Boero2004} establishes the
result for the equiprobable case and proves that those components
are not asymptotically independent in the non-equiprobable case,
thus disputing a claim by Anderson about the generality of his
component tests. They also show that, when 8 cells are used, three
of the four vectors that relate to location, skewness and
kurtosis--as described by Anderson-- are in fact 4-dimensional
rather than 8-dimensional. That is, the same information can be
provided by partitioning the support into 4 cells rather than 8.
Then more vectors with new information come up as we increase the
number of cells (and therefore the dimension) to the next power of
2.

This paper obtains a decomposition in the general case of
non-equiprobable cells with the constraint that each component uses
all cell counts, as in the Hadamard decomposition. This is desirable
because, in regard to higher tests described in the introduction,
using orthogonal vectors with many zeros would yield scarce vectors
of higher order, which in turn give higher order component tests
that use very few counts and are extremely hard to interpret.

Interestingly, it is shown that, when no two cell probabilities are
assumed to be a priori related (i.e. taken as different
indeterminate variables), this approach is only possible when the
number of cells is $4$, $8$ or the trivial case of $2$. This is due
to reasons that pertain to limitations on the possible dimension of
division algebras over the field of real numbers. For cell counts
that are a higher power of $2$, the situation is remedied by using
the theory of orthogonal designs. Indeed, this theory explains why
the cases 2, 4, and 8 are constructible without further assumptions,
and also provides constructions for higher powers by imposing
relations between cell probabilities. To the best of the author's
knowledge, orthogonal designs have not been used in the chi-square
partitioning problem. We will discuss this construction and also
show, via simulation that the power of various components can be
drastically improved by varying the cell probabilities.

\section{The Equiprobable Case}\label{S:prelims}
 The classical goodness-of-fit chi-square test
statistic of Pearson is performed by dividing the range of a
variable into $k$ bins that are mutually exclusive classes, thus
making a multinomial model with respective cell probability vector,
say, ${\bf{p}}=(p_1,\cdots, p_k)$, calculated based on the
hypothesized null distribution. For simplicity of presentation here,
we will assume that the null distribution is completely specified,
so that the above cell probabilities are known exactly and do not
need to be estimated.

We will denote the transpose of a matrix ${\bf A}$ by ${\bf
A^{\prime}}$ and a diagonal matrix by ${\bf D(c)}$, where ${\bf
c}=(c_1,\cdots,c_k)^{\prime}$ and ${\bf D}_{ii}=c_i$. Likewise,
${\bf D}^{1/2}({\bf c})$ denotes a diagonal matrix with diagonal
entries given by $\sqrt{c_i}$ and ${\bf D}^{-1/2}({\bf c})$ is its
inverse when all $c_i$s are distinct from zero.

For a sample of size $n$, the test works by comparing the observed
cell frequency vector ${\bf{m}}=(m_1,\cdots, m_k)$ to the expected
vector ${\bf\mu}=n{\bf{p}}$ via the test statistic
$X^2=\displaystyle\sum_{i=1}^k\frac{(m_i-np_i)^2}{np_i}$ which has
an asymptotic chi-square distribution with $k-1$ degrees of freedom.
This asymptotic distribution is based on the multivariate central
limit theorem that states that the multinomial vector has an
approximate multi-normal distribution with mean $n{\bf p}$ and
variance-covariance matrix with diagonal entries $np_i(1-p_i)$ and
off-diagonal entries ${\bf\Sigma}_{ij}=-np_ip_j$. Thus the
normalized vector

$${\bf x}=\displaystyle\frac{\bf
m-\mu}{\sqrt{n}}=\sqrt{n}\left(\hat{\bf p}-{\bf p}\right)$$

\noindent has an asymptotic multinormal distribution
$\mathcal{N}\left({\bf0},{\bf\Sigma}\right)$, where
${\bf\Sigma}=(\sigma_{ij})$, $\sigma_{ii}=p_i(1-p_i)$ and
$\sigma_{ij}=-p_ip_j$ for $i\neq j$, and $\hat{\bf p}$ is the vector
of relative frequencies. It is, of course, well-known that
${\bf\Sigma}$ is singular with rank $k-1$.

We note that ${\bf\Sigma}$ can be represented as $({\bf
D(p)-pp^{\prime}})$. In the equiprobable case, \textit{i.e.}, when
$p_i=1/k$ for all $i=1,\cdots,k$, ${\bf\Sigma}$ reduces to the form
$\displaystyle\frac1k\left({\bf I}_k-{\bf11^{\prime}}/k\right)$,
where ${\bf I}_k$ is the $k\times k$ identity matrix and ${\bf 1}$
is a $k\times1$ vector of ones.

Since ${\bf\Sigma}$ is positive semi-definite, it can be
diagonalized. If the eigenvalues are
$\lambda_1\leq\lambda_2\leq\cdots\leq\lambda_k$, we can write
\begin{equation}\label{E:diagonalization}
{\bf\Sigma=\Gamma^{\prime}\Lambda\Gamma},
\end{equation}
where
${\bf\Lambda=D}\left((\lambda_1,\cdots,\lambda_k)^{\prime}\right)$
and ${\bf\Gamma}$ is the matrix whose rows are orthonormal
eigenvectors of ${\bf\Sigma}$ corresponding respectively to the
ordered eigenvalues.

Although ${\bf\Sigma}$ is diagonalizable, the presence of a zero
eigenvalue makes it rank deficient. In this case, one would
typically calculate any generalized inverse ${\bf\Sigma^{-}}$ of
${\bf\Sigma}$ and thus use as test statistic the quadratic form
$(\sqrt{n}\left(\hat{\bf p}-{\bf
p}\right)^{\prime}){\bf\Sigma^{-}}(\sqrt{n}\left(\hat{\bf p}-{\bf
p}\right))$ which is asymptotically distributed as $\chi^2_{k-1}$.

Tanabe and Saga \cite{Tanabe1992} provides a treatise on the
symbolic Cholesky decomposition of ${\bf\Sigma}$ and a formula for
its generalized inverse. Watson \cite{Watson96} gives sharp
inequalities for the eigenvalues in terms of the cell probabilities
as well as exact and efficient formulae for calculating the
eigenvectors, which also lead to a formula for the generalized
inverse, under general conditions on the $p_i$s.

In the equiprobable case, Boero \textit{et al} \cite{Boero2004} use
the fact that $\left({\bf I}_k-{\bf 11}^{\prime}/k\right)$ is an
idempotent matrix (and hence is a weak inverse of itself) to
represent the Pearson statistic as
\begin{equation}\label{E:chi-square}
X^2=(\sqrt{n}\left(\hat{\bf p}-{\bf p}\right)^{\prime}){\left({\bf
I}_k-{\bf11^{\prime}}/k\right)}(\sqrt{n}\left(\hat{\bf p}-{\bf
p}\right))/(1/k).
\end{equation}

In the case $k=2^l$, they use a Hadamard matrix (without the
constant row) to produce $(k-1)$ asymptotically independent
components of $X^2$, giving a theoretical basis to the seemingly
\textit{ad hoc} selection of vectors in Anderson
\cite{Anderson1994}. Despite the elegance, Boero \textit{et al}
\cite{Boero2004} does not justify the use of Hadamard matrices in
the above procedure. In the following paragraphs, we will explain
their inherent relevance in the diagonalization, thus setting the
stage for a generalization of this procedure to the non-equiprobable
case.

In effect, we note that the constant vector ${\bf1}$ is an
eigenvector of ${\bf\Sigma}$ that corresponds to the simple
eigenvalue $0$. Since ${\bf\Sigma}$ is symmetric and idempotent, it
admits only the eigenvalues $0$ and $1$. It follows that any vector
orthogonal to ${\bf1}$ is an eigenvector of $\lambda=1$. Since all
rows of a Hadamard matrix $H$ are orthogonal with constant norm
$\sqrt{k}$, the diagonalization (\ref{E:diagonalization}) is
satisfied with ${\bf\Gamma}=\displaystyle\frac1{\sqrt{k}}{\bf H}$
(where ${\bf H}$ is written in the standard form so that the first
row of ${\bf H}$ is the vector ${\bf1}^{\prime}$), and
${\bf\Lambda}$ coincides with $\displaystyle\frac1k{\bf I}_k$ except
for ${\bf\Lambda}_{11}=0$.

If we now apply the transformation ${\bf
y}=\displaystyle\frac1{\sqrt{k}}{\bf H x}$, it follows from the
general theory that ${\bf y}$ converges to a centered multinormal
distribution with variance-covariance matrix  ${\bf\Lambda}$. This
and Equation~(\ref{E:chi-square}) immediately lead to the
decomposition $X^2={\bf y^{\prime}y}/(1/k)$.

\section{The Non-Equiprobable Case}\label{S:nonequiprobable}
In this section we present a generalization to the above procedure
in the non-equiprobable case, keeping the number of classes a power
of $2$. We remark here that the diagonalization in the equiprobable
case of the previous section rests on the fact that $k{\bf\Sigma}$
is idempotent. This property is lost in the general case. While $0$
is always a simple eigenvalue, ${\bf\Sigma}$ admits in general many
other eigenvalues that depend on the cell probability vector ${\bf
p}$. Indeed, Watson \cite{Watson96} establishes that, if $p_i\leq
p_{i+1}$ for all $i$ then each of the $k-1$ nonzero eigenvalues is
squeezed between two consecutive $p_is$. i.e., $p_1\leq\lambda_1\leq
p_2\leq\cdots\leq p_{k-1}\leq\lambda_{k-1}\leq p_k$. This
immediately establishes that ${\bf\Sigma}$ is idempotent if and only
if all the $p_is$ are equal. The results in this section are
outlined as follows. We first propose an appropriate transformation
to the vector ${\bf x}$ that will produce a vector whose
variance-covariance matrix is idempotent. We later proceed by
providing a matrix of highly structured eigenvectors, symbolically
defined in terms of ${\bf p}$, and we characterize the cases when
this matrix is orthogonal, which will then readily decompose the
Pearson statistic.

\begin{theorem}\label{T:idemp-transformation}
Let ${\bf y}={\bf D}^{-1/2}({\bf p}){\bf x}$. Then ${\bf y}$
converges in distribution to a multinormal distribution with an
idempotent variance-covariance matrix ${\bf\Sigma^{*}}$.
Furthermore, the vector $\sqrt{\bf p}$ defined as
$\left(\sqrt{p_1},\cdots,\sqrt{p_k}\right)^{\prime}$ is an
eigenvector that spans the kernel of ${\bf\Sigma^{*}}$.
\end{theorem}

It is worth noting that, componentwise, this transformation is a
re-scaling of the already normalized multinomial frequencies.
Namely, the coordinates of ${\bf y}$ are
$y_i=\displaystyle\frac{x_i}{\sqrt{p_i}}=\displaystyle\frac{m_i-np_i}{\sqrt{np_i}}$.

In the following Lemma, we will gather some simple facts that will
be used in the proof of Theorem~\ref{T:idemp-transformation} and
that can be checked by direct calculations.
\begin{lemma}\label{L:calculations}
The following are true for any probability vector ${\bf p}$.
    \begin{enumerate}
    \item[(a)] ${\bf D}^{-1}({\bf p}){\bf p}={\bf 1}$.
    \item[(b)] ${\bf p^{\prime}\; D}^{-1}({\bf p}){\bf p}$ is the scalar value
    $1$.
    \item[(c)] ${\bf D}^{-1/2}({\bf p})\sqrt{\bf p}={\bf 1}$
    \item[(d)] ${\bf D}^{-1/2}({\bf p}){\bf p}=\sqrt{\bf p}$
    \item[(e)] ${\bf D}^{1/2}({\bf p}){\bf 1}=\sqrt{\bf p}$
    \end{enumerate}
\end{lemma}
\begin{proof}[Proof of Theorem \ref{T:idemp-transformation}]
We first prove that ${\bf\Sigma^{*}}$ is idempotent by direct
multiplication. Observing that ${\bf\Sigma^{*}}={\bf D}^{-1/2}({\bf
p}){\bf\Sigma}{\bf D}^{-1/2}({\bf p})$, we get
\begin{eqnarray*}
{\bf\Sigma^{*}}^2&=&{\bf D}^{-1/2}({\bf p})({\bf
D(p)-pp^{\prime}}){\bf D}^{-1/2}({\bf p})\times {\bf D}^{-1/2}({\bf
p})({\bf D(p)-pp^{\prime}}){\bf D}^{-1/2}({\bf p})\\
&=&{\bf D}^{-1/2}({\bf p}){\bf C}{\bf D}^{-1/2}({\bf p})
\end{eqnarray*}
where
\[
{\bf C}={\bf D}{({\bf p})}-2{\bf pp^{\prime}}+{\bf
pp^{\prime}D}^{-1}({\bf p}){\bf pp^{\prime}}.
\]
The last expression in the last row is ${\bf pp^{\prime}}$ by
Lemma~\ref{L:calculations}(b). This shows that
${\bf\Sigma^{*}}^2={\bf\Sigma^{*}}$.

Next, consider ${\bf\Sigma^{*}}\sqrt{\bf p}={\bf D}^{-1/2}({\bf
p})({\bf D(p)-pp^{\prime}}){\bf 1}={\bf D}^{1/2}({\bf p}){\bf
1}-{\bf D}^{-1/2}({\bf p}){\bf p}({\bf p^{\prime} 1})$\newline

$={\bf D}^{1/2}({\bf p}){\bf 1}-{\bf D}^{-1/2}({\bf p}){\bf
p}=\sqrt{\bf p}-\sqrt{\bf p}={\bf 0}$. Where we have used parts (c),
(d) and (e) of Lemma~\ref{L:calculations}.
\end{proof}

The fact that $\sqrt{\bf p}$ is the only eigenvector of $0$ plays a
role similar to the vector ${\bf 1}$ in the equiprobable case as
outlined above. The next section shows how to construct an
orthogonal matrix whose first column is $\sqrt{\bf p}$ and,
therefore, whose other columns are eigenvectors of ${\bf\Sigma^{*}}$
all with eigenvalue $1$. For now, suppose that ${\bf O}$ is an
orthogonal matrix with $\sqrt{\bf p}$ as its first column. On one
hand, ${\bf\Sigma^{*}}={\bf O^{\prime}\Lambda O}$, where
${\bf\Lambda}$ coincides with ${\bf I}_k$ except for
${\bf\Lambda}_{11}=1$. Letting ${\bf A}$ be the submatrix of ${\bf
O}$ obtained by omitting the first column, it is immediate that
${\bf A^{\prime}A}={\bf I}_{k-1}$ and ${\bf A
A^{\prime}}={\bf\Sigma}^*$. Thus, ${\bf y^{\prime}\Sigma^*y}={\bf
y^{\prime}AA^{\prime}y}={\bf(A^{\prime}y)^{\prime}(A^{\prime}y)}$.
On the other hand, with ${\bf D(p)}={\bf D}$ and ${\bf D}^{1/2}{\bf
(p)}={\bf D}^{1/2}$,
\begin{eqnarray*}
{\bf y^{\prime}\Sigma^*y}
&=&{\left({\bf D}^{-1/2}(\frac{\bf m-\mu}{\sqrt{n}})\right)^{\prime}{\bf D}^{-1/2}\left({\bf D-pp}^{\prime}\right){\bf D}^{-1/2}\left({\bf D}^{-1/2}(\frac{\bf m-\mu}{\sqrt{n}})\right)}\\
&=&\frac1n({\bf m-\mu})^{\prime}{\bf D}^{-1/2}\left({\bf
D}^{-1/2}{\bf\Sigma}{\bf D}^{-1/2}\right){\bf D}^{-1/2}({\bf
m-\mu})\\
&=&\frac1n({\bf m-\mu})^{\prime}\left({\bf D}^{-1}-{\bf D}^{-1}{\bf
pp}^{\prime}{\bf D}^{-1}\right)({\bf m-\mu})\\
&=&\frac1n({\bf m-\mu})^{\prime}{\bf D}^{-1}({\bf
m-\mu})-\frac1n({\bf m-\mu})^{\prime}{\bf D}^{-1}{\bf
pp}^{\prime}{\bf D}^{-1}({\bf m-\mu}).\\
\end{eqnarray*}
The first expression in the last row is clearly the Pearson
statistic $X^2$. By Part~(a) of Lemma~\ref{L:calculations}, $({\bf
m-\mu})^{\prime}{\bf D}^{-1}{\bf p}$ is the zero vector, so that the
second term collapses to zero. Hence we have obtained the partition
$X^2={\bf(A^{\prime}y)^{\prime}(A^{\prime}y)}$. If we denote the
columns of $A$ by ${\bf v}_2\cdots,{\bf v}_k$, this partition can be
expressed as
\begin{equation}\label{E:partition}
X^2=\sum_{l=2}^k\left({\bf v}_l^{\prime}{\bf y}\right)^2,
\end{equation}
\noindent where ${\bf v}_l^{\prime}{\bf
y}=\displaystyle\sum_{i=1}v_{li}\left(\displaystyle\frac{m_i-np_i}{\sqrt{np_i}}\right)$
and $v_{li}$ is the $i^{th}$ coordinate of ${\bf v}_l$.

\section{A Latin-Hadamard Matrix}\label{S:Latin Hadamard}
As planned above, in this section we construct a matrix of
orthogonal eigenvectors whose first column is $\sqrt{\bf p}$, the
generator of the kernel of ${\bf\Sigma^{*}}$. To this end, we will
first construct a square matrix with integer entries, whose special
features will allow the construction we propose.

For $k=2^w$, define a matrix ${\bf S}=[{\bf c}_1,\cdots, {\bf c}_k]$
with columns ${\bf c}_i$ such that ${\bf
c}_1=[1,\cdots,k]^{\prime}$, and for $i=1,\cdots,2^{w-1}$ let
\[
c_{2i,2}=c_{2i-1,1};\;\textup{ and } c_{2i-1,2}=c_{2i,1};
\]
Now for $l=0,\cdots,w-1$, write ${\bf S}$ in blocks of size
$2^l\times2^l$ matrices as
\[
\left[
\begin{tabular}{lll}
${\bf S}^l_{1,1}$       & $\cdots$  & ${\bf S}^l_{1,2^{w-l}}$\\
$\vdots$                & $\ddots$  & $\vdots$\\
${\bf S}^l_{2^{w-l},1}$ & $\cdots$  & ${\bf S}^l_{2^{w-l},2^{w-l}}$
\end{tabular}
\right]
\]
The first block column ${\bf S}^1_{\cdot,1}$ of ${\bf S}$ has been
defined. We fill in the rest of the matrix as follows. For each
value of $l$, $l=2,\cdots,w-1$, we define the second column ${\bf
S}^l_{\cdot,2}$ of block matrices as
\[
{\bf S}^l_{2i,2}={\bf S}^l_{2i-1,1};\;\textup{ and } {\bf
S}^l_{2i-1,2}={\bf S}^l_{2i,1};
\]
\noindent for each $i=1,\cdots,2^{w-l}$. We note here that the
second column can be defined in several ways, producing the same
matrix up to a permutation of the rows. For example, the second
column ${\bf c}_2$ can be defined as $[k,\cdots,1]^{\prime}$. The
current construction is most convenient for recursive arguments due
to its self-similarities, as apparent in Figure~\ref{F:Latin16},
where the upper left $2^w\times2^w$ sub-matrix is one such matrix
for $w=1,\cdots,4$.

\begin{figure}[h!]
\center $\left(
\begin{tabular}{cccccccccccccccc}
1 & 2 & 3 & 4 & 5 & 6 & 7 & 8 & 9 & 10 & 11 & 12 & 13 & 14  & 15 &  16\\
2 & 1 & 4 & 3 & 6 & 5 & 8 & 7 & 10 & 9 & 12 & 11 & 14 & 13  & 16 &  15\\
3 & 4 & 1 & 2 & 7 & 8 & 5 & 6 & 11 & 12 & 9 & 10 & 15 & 16  & 13 &  14\\
4 & 3 & 2 & 1 & 8 & 7 & 6 & 5 & 12 & 11 & 10 & 9 & 16 & 15  & 14 &  13\\
5 & 6 & 7 & 8 & 1 & 2 & 3 & 4 & 13 & 14 & 15 & 16 & 9 & 10  & 11 &  12\\
6 & 5 & 8 & 7 & 2 & 1 & 4 & 3 & 14 & 13 & 16 & 15 & 10 & 9  & 12 &  11\\
7 & 8 & 5 & 6 & 3 & 4 & 1 & 2 & 15 & 16 & 13 & 14 & 11 & 12 & 9  &  10\\
8 & 7 & 6 & 5 & 4 & 3 & 2 & 1 & 16 & 15 & 14 & 13 & 12 & 11 & 10 &  9\\
9 & 10 & 11 & 12 & 13 & 14  & 15 &  16 & 1 & 2 & 3 & 4 & 5 & 6 & 7 & 8\\
10 & 9 & 12 & 11 & 14 & 13  & 16 &  15 & 2 & 1 & 4 & 3 & 6 & 5 & 8 & 7\\
11 & 12 & 9 & 10 & 15 & 16  & 13 &  14 & 3 & 4 & 1 & 2 & 7 & 8 & 5 & 6\\
12 & 11 & 10 & 9 & 16 & 15  & 14 &  13 & 4 & 3 & 2 & 1 & 8 & 7 & 6 & 5\\
13 & 14 & 15 & 16 & 9 & 10  & 11 &  12 & 5 & 6 & 7 & 8 & 1 & 2 & 3 & 4\\
14 & 13 & 16 & 15 & 10 & 9  & 12 &  11 & 6 & 5 & 8 & 7 & 2 & 1 & 4 & 3\\
15 & 16 & 13 & 14 & 11 & 12 & 9  &  10 & 7 & 8 & 5 & 6 & 3 & 4 & 1 & 2\\
16 & 15 & 14 & 13 & 12 & 11 & 10 &  9 & 8 & 7 & 6 & 5 & 4 & 3 & 2  &
1\\
\end{tabular}
\right)$ \caption{A $2^4\times2^4$ Latin square with the AB-BA
property.}\label{F:Latin16}
\end{figure}

Several immediate properties that can be checked by induction make
this matrix special. We first list a few such properties in the form
of a lemma, then we state and prove a proposition that deals with
the special feature that will allow us to turn ${\bf S}$ into an
orthogonal matrix by multiplying some entries by $-1$.

\begin{lemma}\label{L:Latin}
The following assertions are true about ${\bf S}$.

(a) ${\bf S}$ is a Latin square.

(b)  ${\bf S}$ is symmetric with respect to both the first and
second diagonals.

(c) ${\bf S}$ can be defined recursively in $w$ as follows. For
$w=0$, let ${\bf S}_0=1$. For an integer $w\geq1$ let ${\bf
S}_w=\left[
\begin{tabular}{ll}
${\bf A}$ & ${\bf B}$\\
${\bf B}$ & ${\bf A}$
\end{tabular}
\right]$, where ${\bf A=S}_{w-1}$ and  ${\bf B}={\bf A}+2^{w-1}{\bf
1}_{2^{w-1}}{\bf 1}^{\prime}_{2^{w-1}}$, and ${\bf 1}_{2^{w-1}}$ is
the $2^{w-1}$-dimensional constant vector of ones.
\end{lemma}

\begin{proposition}\label{P:rectangleProperty}
For dimension $k=2^w$ and indices $1\leq i_1,j_1,j_2\leq k$,
$j_1\neq j_2$, if ${\bf S}_{i_1,j_1}=a$ and ${\bf S}_{i_1,j_2}=b$,
then there exist an index $i_2\neq i_1$, $1\leq i_2\leq k$ such that
${\bf S}_{i_2,j_1}=b$ and ${\bf S}_{i_2,j_2}=a$.
\end{proposition}

\begin{proof} Let ${\bf S}=\left[
\begin{tabular}{ll}
${\bf A}$ & ${\bf B}$\\
${\bf B}$ & ${\bf A}$
\end{tabular}
\right]$ as in Lemma~(\ref{L:Latin}c). We will use an inductive
argument. So, it suffices to assume that $a$ and $b$ are entries of
${\bf A}$ and ${\bf B}$ respectively. That is, for arbitrary indices
$s_1$, $t_1$, and $t_2$, $1\leq s_1,t_1,t_2\leq2^{w-1}$, let ${\bf
A}_{s_1,t_1}=a$ and ${\bf B}_{s_1,t_2}=b$, where $t_1$ and $t_2$
need not be different. Since ${\bf S}$ is a Latin square, there
exists $s_2$ ($1\leq s_2\leq2^{w-1}$) such that ${\bf
S}_{s_2+2^{w-1},t_1}=b$. Referring to the block matrix composition
of ${\bf S}$, this means that ${\bf B}_{s_2,t_1}=b$. Since the
matrices ${\bf A}$ and ${\bf B}$ have the same structure and
relative positions, by Lemma~(\ref{L:Latin}c), it follows that ${\bf
A}_{s_1,t_2}={\bf A}_{s_2,t_1}=b^{\prime}$ for some $b^{\prime}$,
$1\leq b^{\prime}\leq2^{w-1}$. Applying the inductive hypothesis to
the matrix ${\bf A}$, it is evident that ${\bf A}_{s_2,t_2}={\bf
A}_{s_1,t_1}=a$. To formally complete the proof we take $i_1=s_1$,
$j_1=t_1$ and $i_2=s_2+2^{w-1}$, $j_2=t_2+2^{w-1}$.

\end{proof}

We will refer to the property proven in the above proposition as the
AB-BA property, and to every four entries that make up this property
as AB-BA corners. This property potentially allows us, by
appropriately inserting minus signs to certain entries, to come up
with a matrix that has orthogonal columns. Namely, if among every
four AB-BA corners exactly three are `colored' by the same color
(where the two available colors are `+' and `-') then all columns
will be orthogonal. It turns out that this is possible only in few
cases as stated in the following theorem.

\begin{theorem}\label{T:LatinHadamard}
When $k=2$, $4$ or $8$, there exists a matrix ${\bf H}$ with
orthogonal columns and orthogonal rows and that has the following
properties

(a) $|{\bf H}_{ij}|={\bf S}_{ij}$,

(b) The matrix ${\bf \tilde{H}}$ defined componentwise by ${\bf
\tilde{H}}_{ij}=sgn({\bf H}_{ij})$ is a Hadamard matrix.

(c) The signs of entries of ${\bf H}$ can be arranged so that all
entries on the first row and first column are positive.

Furthermore, an orthogonal matrix with these properties is not
possible for any other value of $k$.
\end{theorem}

\begin{proof}
In an appendix at the end of this paper, we  we present a systematic
way to \textit{color} the entries by `+' and `-' in a way that
guarantees that every column is orthogonal to the first column and
every row is orthogonal to the first row. As it turns out, there are
$2^{2^{w-1}-w}$ ways to make this happen. When $w=1,2,3$ each of
these colorings yields mutually orthogonal columns. (We will refer
to such a matrix with orthogonal columns as a Latin-Hadamard
matrix.) When $k=4$, no such coloring makes all the columns mutually
orthogonal. As the coloring is performed by first coloring the upper
left quarter of the matrix, this already shows that no coloring is
possible for $k\geq4$. Now to see that there does not exist any
colored matrix with $k=4$ suppose one exists. Then necessarily the
transpose is a colored matrix. In particular, every column is
orthogonal to the first column and every row is orthogonal to the
first row. Our exhaustive search shows that no such matrix with the
latter property has all of its columns mutually orthogonal.
\end{proof}

The above proof is a brute force argument, which does not reveal why
the so-called Latin-Hadamard matrices do not exist in higher
dimensions. In fact, the columns are all orthogonal if the equation
\begin{equation}\label{E:OrthoConditions}
\alpha_{i,j}\cdot\alpha_{i,j+l^{\prime}}=-\alpha_{i+l,j}\cdot\alpha_{i+l,j+l^{\prime}}
\end{equation}
is satisfied for all $i$ and $j$ between 1 and $2^k$, $l\neq0$ and
$l^{\prime}\neq0$ such that
$\textbf{S}_{i,j}=\textbf{S}_{i+l,j+l^{\prime}}$ and
$\textbf{S}_{i,j+l^{\prime}}=\textbf{S}_{i+l,j}$, and where the
values of the $\alpha$s are defined via
$\textbf{H}_{ij}=\alpha_{ij}\textbf{S}_{ij}$.

Indeed, the coloring scheme--described in the appendix--is performed
with the only objective of making Equation~\ref{E:OrthoConditions}
work for $i=1$ and $j\neq1$ or $i\neq1$ and $j=1$. That is, it only
ensures that each column is orthogonal to the first column and each
row is orthogonal to the first row. There is no guarantee that the
numerous other equations of the form (\ref{E:OrthoConditions}) for
$i\neq1$ and $j\neq1$ are satisfied. For $k=1,2,3$ a solution comes
as a free gift. That is, we found a solution of an over-determined
system by only considering some of the equations. In
Section~\ref{S:DivisionAlgebra} below we give another proof of
Theorem~\ref{T:LatinHadamard}, that better explains why the coloring
scheme works only for $k=1,2,3$. In the rest of this section we will
describe how each colored matrix corresponding to $w\leq3$ can be
turned into an orthogonal matrix $\bf{O}$ that partitions the
chi-square statistic in a way that each component test is a sum of
differences of two cell frequencies in some order, where each cell
frequency is used exactly once.

Figure~\ref{F:LatHad8by8} gives a matrix ${\bf H}$ of size $8$,
which also gives examples of lower sizes given by the upper left
$2\times2$ and $4\times4$ sub-matrices respectively.

Interestingly, when $k=16$, many pairs of vectors produced by our
current construction are already  orthogonal. Indeed the first eight
columns (as well as the other eight) are all mutually orthogonal.
Since different vectors are typically used as individual components,
one is still able to use many of these vectors as independent
component tests that can individually provide useful information.

\begin{figure}[h!]
\center $\left(
\begin{tabular}{cccccccc}
+1 & +2 & +3 & +4 & +5 & +6 & +7 & +8\\
+2 & -1 & -4 & +3 & +6 & -5 & +8 & -7\\
+3 & +4 & -1 & -2 & +7 & -8 & -5 & +6\\
+4 & -3 & +2 & -1 & +8 & +7 & -6 & -5\\
+5 & -6 & -7 & -8 & -1 & +2 & +3 & +4\\
+6 & +5 & +8 & -7 & -2 & -1 & +4 & -3\\
+7 & -8 & +5 & +6 & -3 & -4 & -1 & +2\\
+8 & +7 & -6 & +5 & -4 & +3 & -2 & -1\\
\end{tabular}
\right)$ \caption{A $8\times8$ signed Latin square with orthogonal
columns and rows.}\label{F:LatHad8by8}
\end{figure}

Evidently, one can substitute the integers $1,\cdots,2^w$, in the
matrix ${\bf H}$, respectively with any set of $2^w$ real numbers
and keep the orthogonality of the columns, thanks to the AB-BA
property of the underlying matrix ${\bf S}$. Returning to the main
problem of this paper, we are now ready to give an orthonormal
matrix of eigenvectors of ${\bf\Sigma}^{*}$ when the number of
categories in the multinomial model is a $2$, $4$ or $8$. We present
this matrix in the following corollary to
Theorem~\ref{T:LatinHadamard}.

\begin{corollary}
Consider the matrix ${\bf O}$ defined componentwise by the equation
$${\bf O}_{ij}=sgn({\bf H}_{ij})\cdot\sqrt{p_s}$$ where $s=|{{\bf
H}_{ij}}|$, and $p_s;\;s=1,\cdots,2^w$ are the components of the
vector $\sqrt{\bf p}$ and ${\bf H}$ is as given in
Theorem~\ref{T:LatinHadamard}(c). Then the columns of ${\bf O}$ form
a basis of orthonormal eigenvectors of ${\bf\Sigma^*}$, such that
the first column is the zero eigenvector $\sqrt{\bf p}$.
\end{corollary}

All $4\times4$ and $8\times8$ Latin-Hadamard matrices are displayed
in Table~\ref{T: LatHad4AND8}.

\section{Power Simulation for Components}\label{S:PowSim}
As an illustration, we will use the matrix on the left of the first
row of $8\times8$ matrices in Table~\ref{T: LatHad4AND8}, which
displays all possible Latin-Hadamard matrices of dimensions
$4\times4$ and $8\times8$. Since the columns of this matrix are
orthogonal, the components of the Pearson statistic, given by the
summands in Equation~(\ref{E:partition}), are asymptotically
independent chi-square statistics, each with one degree of freedom.
Furthermore, each component can be used individually as a separate
test. Since the first column is an eigenvector of the zero
eigenvalue, it contains no information. Using the second column, we
get the component test, which we will call $T_2$,

\begin{eqnarray*}
&&\frac1n\{\sqrt{\frac{p_2}{p_1}}\left(m_1-np_1\right)-\sqrt{\frac{p_1}{p_2}}\left(m_2-np_2\right)
   +\sqrt{\frac{p_4}{p_3}}\left(m_3-np_3\right)-\sqrt{\frac{p_3}{p_4}}\left(m_4-np_4\right)\\
   &&+\sqrt{\frac{p_6}{p_5}}\left(m_5-np_5\right)-\sqrt{\frac{p_5}{p_6}}\left(m_6-np_6\right)
   +\sqrt{\frac{p_8}{p_7}}\left(m_7-np_7\right)-\sqrt{\frac{p_7}{p_8}}\left(m_8-np_8\right)\}^2\\
\end{eqnarray*}
Noting that $T_2$, the square root of the above component test, is
asymptotically normal, and using
$\hat{p}_i=\displaystyle\frac{m_i}n$, the statistic $T_2$ reduces as
follows

\begin{eqnarray*}
\sqrt{T_2}&=&\sqrt{n}\left(\left(\sqrt{\frac{p_2}{p_1}}
\hat{p}_1-\sqrt{\frac{p_1}{p_2}}\hat{p}_2\right)
   +\left(\sqrt{\frac{p_4}{p_3}}\hat{p}_3-\sqrt{\frac{p_3}{p_4}}\hat{p}_4\right)\right.\\
   &+&\left(\sqrt{\frac{p_6}{p_5}}\hat{p}_5-\sqrt{\frac{p_5}{p_6}}\hat{p}_6\right)
   +\left.\left(\sqrt{\frac{p_8}{p_7}}\hat{p}_7-\sqrt{\frac{p_7}{p_8}}\hat{p}_8\right)\right).\\
\end{eqnarray*}

Similarly, the sixth and eighth vectors yield the asymptotically
normal component tests
\begin{eqnarray*}
T_6 &=&\sqrt{n}\left(\left(\sqrt{\frac{p_6}{p_1}}
\hat{p}_1-\sqrt{\frac{p_1}{p_6}}\hat{p}_6\right)
   +\left(\sqrt{\frac{p_5}{p_2}}\hat{p}_2-\sqrt{\frac{p_2}{p_5}}\hat{p}_5\right)\right.\\
   &-&\left(\sqrt{\frac{p_8}{p_3}}\hat{p}_3-\sqrt{\frac{p_3}{p_8}}\hat{p}_8\right)
   +\left.\left(\sqrt{\frac{p_7}{p_4}}\hat{p}_4-\sqrt{\frac{p_4}{p_7}}\hat{p}_7\right)\right).\\
\end{eqnarray*}
and
\begin{eqnarray*}
T_8&=&\sqrt{n}\left(\left(\sqrt{\frac{p_8}{p_1}}
\hat{p}_1-\sqrt{\frac{p_1}{p_8}}\hat{p}_8\right)
   -\left(\sqrt{\frac{p_7}{p_2}}\hat{p}_2-\sqrt{\frac{p_2}{p_7}}\hat{p}_7\right)\right.\\
   &+&\left(\sqrt{\frac{p_6}{p_3}}\hat{p}_3-\sqrt{\frac{p_3}{p_6}}\hat{p}_6\right)
   +\left.\left(\sqrt{\frac{p_5}{p_4}}\hat{p}_4-\sqrt{\frac{p_4}{p_5}}\hat{p}_5\right)\right).\\
\end{eqnarray*}

Notice that the mutual weighted differences between pairs of sample
proportions make it relatively easier to interpret many individual
components. Suppose a null hypothesis specifies a distribution $F$
and subdivide its support into $8$ bins with probabilities
$p_1,\cdots,p_8$ labeled from left to right. Generally, the first
three vectors (second, third and fourth) deal with local changes in
the distribution while the other four vectors compare bins from the
opposite sides of the support. In particular, in the examples above,
$T_2$ compares consecutive bins, so it is sensitive to the rate at
which a distribution increases or decays; while $T_8$ is sensitive
to a change in location.

To illustrate the significance of this partition into component
tests, the table below shows results of power simulation when the
null distribution is assumed to be $N(0,1)$. The actual data is
generated from several distributions with a change from standard
normal along a location shift, an increase in variance or an
increase in the tail. From each distribution in the table, $10000$
samples of size $n=200$ were generated from the distribution named
in the first column. The real axis was divided into eight
sub-domains whose standard normal probabilities from left to right
are equal to the specified multinomial vector. The table estimates
the power of the overall chi-square statistic as well as all the
components provided by the columns of the upper left matrix in
Table~\ref{T: LatHad4AND8}. It is worth noting that various
components behave drastically differently. The middle two rows (in
the upper part of the table) investigate the power with respect to
location shift. Components $T_4$ and $T_8$ are the only components
that are sensitive to this location shift, with $T_8$ outperforming
the $X^2_P$. One can also compare the power of the equally likely
case (Case (a)) to the other cases. The most important remark is
that the choice of the multinomial vector directly affects the power
of various components; notice how Case (c) allows all components to
pick up some nontrivial power. The lower part of the table compares
gamma generated values with given shape and scale parameters to a
null normal model with matched mean and standard deviation. This
case is meant to check which component is sensitive to skewness and
it clearly shows the superiority of each component test to the
global statistic $X^2_P$. Components $T_3$ and $T_5$ consistently
show a higher sensitivity. Most remarkably, changing the probability
vector can greatly increase the power of the same component test.
For example, Component $T_4$ becomes sensitive to skewness when the
probability vector is not uniform. Similarly, $T_8$ picks up a lot
of power against fat tails (Cauchy and t distribution) when the
probability vector (c) is used. This is also observed with
$N(0,1.3)$ data.

\begin{table}[h]
\begin{tabular}{c||c|c|c|c|c|c|c|c|c}\hline
    && $X^2_P$ & $T_2$ & $T_3$ & $T_4$ & $T_5$ &  $T_6$ &  $T_7$ &  $T_8$\\\hline
               & (a) & $0.859$ & $0.057$ & $0.057$ & $0.059$ & $0.629$ & $0.850$ & $0.524$ & $0.056$\\
    $N(0,1.3)$ & (b) & $0.964$ & $0.063$ & $0.088$ & $0.096$ & $0.763$ & $0.977$ & $0.329$ & $0.055$\\
               & (c) & $0.878$ & $0.052$ & $0.062$ & $0.289$ & $0.533$ & $0.788$ & $0.425$ & $0.405$\\\hline
               & (a) & $0.525$ & $0.058$ & $0.052$ & $0.058$ & $0.353$ & $0.577$ & $0.287$ & $0.053$\\
    $N(0,1.2)$ & (b) & $0.718$ & $0.058$ & $0.076$ & $0.078$ & $0.450$ & $0.812$ & $0.160$ & $0.055$\\
               & (c) & $0.559$ & $0.052$ & $0.059$ & $0.156$ & $0.287$ & $0.494$ & $0.221$ & $0.230$\\\hline
               & (a) & $0.478$ & $0.136$ & $0.068$ & $0.286$ & $0.059$ & $0.067$ & $0.059$ & $0.630$\\
    $N(0.2,1)$ & (b) & $0.489$ & $0.133$ & $0.058$ & $0.288$ & $0.056$ & $0.066$ & $0.048$ & $0.628$\\
               & (c) & $0.428$ & $0.133$ & $0.054$ & $0.181$ & $0.106$ & $0.152$ & $0.073$ & $0.494$\\\hline
               & (a) & $0.992$ & $0.371$ & $0.118$ & $0.789$ & $0.102$ & $0.166$ & $0.087$ & $0.994$\\
    $N(0.4,1)$ & (b) & $0.994$ & $0.385$ & $0.064$ & $0.817$ & $0.118$ & $0.237$ & $0.054$ & $0.996$\\
               & (c) & $0.990$ & $0.422$ & $0.091$ & $0.563$ & $0.358$ & $0.603$ & $0.212$ & $0.968$\\\hline
               & (a) & $0.798$ & $0.054$ & $0.055$ & $0.056$ & $0.632$ & $0.531$ & $0.727$ & $0.055$\\
    $t(2)$     & (b) & $0.996$ & $0.067$ & $0.089$ & $0.111$ & $0.979$ & $0.952$ & $0.925$ & $0.061$\\
               & (c) & $0.935$ & $0.086$ & $0.259$ & $0.049$ & $0.577$ & $0.666$ & $0.791$ & $0.355$\\\hline
               & (a) & $0.999$ & $0.053$ & $0.060$ & $0.056$ & $0.980$ & $0.942$ & $0.994$ & $0.055$\\
    $t(1)$     & (b) & $1.000$ & $0.083$ & $0.118$ & $0.150$ & $1.000$ & $1.000$ & $1.000$ & $0.054$\\
               & (c) & $0.999$ & $0.162$ & $0.593$ & $0.855$ & $0.969$ & $0.987$ & $0.996$ & $0.760$\\\hline\hline
               & (a) & $0.524$ & $0.055$ & $0.315$ & $0.082$ & $0.586$ & $0.158$ & $0.143$ & $0.088$\\
$gamma(5,1/5)$ & (b) & $0.859$ & $0.067$ & $0.561$ & $0.548$ & $0.673$ & $0.088$ & $0.061$ & $0.078$\\
               & (c) & $0.771$ & $0.045$ & $0.313$ & $0.337$ & $0.802$ & $0.101$ & $0.039$ & $0.091$
               \\\hline
               & (a) & $0.234$ & $0.055$ & $0.153$ & $0.063$ & $0.307$ & $0.122$ & $0.096$ & $0.064$\\
$gamma(10,1)$ & (b)  & $0.445$ & $0.063$ & $0.314$ & $0.273$ & $0.350$ & $0.097$ & $0.059$ & $0.068$\\
               & (c) & $0.335$ & $0.046$ & $0.158$ & $0.153$ & $0.480$ & $0.059$ & $0.040$ & $0.069$\\\hline
\end{tabular}
\caption{The values in each cell correspond, respectively from
above, to the simulated power when the multinomial probability
vector is proportional to (a) $(1,1,1,1,1,1,1,1)$, (b)
$(1,2,3,4,4,3,2,1)$, (c) $(1,2,3,4,1,2,3,4)$.}\label{T: PowerSim}
\end{table}

\begin{table}
\center
\begin{tabular}{cc} \tiny $\left(
\begin{tabular}{cccc}
$+\sqrt{p_1}$ &  $+\sqrt{p_2}$ &  $+\sqrt{p_3}$ &  $+\sqrt{p_4}$\\
$+\sqrt{p_2}$ &  $-\sqrt{p_1}$ &  $+\sqrt{p_4}$ &  $-\sqrt{p_3}$\\
$+\sqrt{p_3}$ &  $-\sqrt{p_4}$ &  $-\sqrt{p_1}$ &  $+\sqrt{p_2}$\\
$+\sqrt{p_4}$ &  $+\sqrt{p_3}$ &  $-\sqrt{p_2}$ &  $-\sqrt{p_1}$\\
\end{tabular}
\right)$ & \tiny $\left(
\begin{tabular}{cccc}
$+\sqrt{p_1}$ &  $+\sqrt{p_2}$ &  $+\sqrt{p_3}$ &  $+\sqrt{p_4}$\\
$+\sqrt{p_2}$ &  $-\sqrt{p_1}$ &  $-\sqrt{p_4}$ &  $+\sqrt{p_3}$\\
$+\sqrt{p_3}$ &  $+\sqrt{p_4}$ &  $-\sqrt{p_1}$ &  $-\sqrt{p_2}$\\
$+\sqrt{p_4}$ &  $-\sqrt{p_3}$ &  $+\sqrt{p_2}$ &  $-\sqrt{p_1}$\\
\end{tabular}
\right)$
\end{tabular}\vspace{12pt}
\begin{tabular}{cc} \hspace{-12pt}\tiny $\left(
\begin{tabular}{cccccccc}%
$+\check{p}_1$ &  $+\check{p}_2$ &  $+\check{p}_3$ &  $+\check{p}_4$ & $+\check{p}_5$ & $+\check{p}_6$ &  $+\check{p}_7$ &  $+\check{p}_8$\\
$+\check{p}_2$ &  $-\check{p}_1$ &  $-\check{p}_4$ &  $+\check{p}_3$ & $-\check{p}_6$ & $+\check{p}_5$ &  $-\check{p}_8$ &  $+\check{p}_7$\\
$+\check{p}_3$ &  $+\check{p}_4$ &  $-\check{p}_1$ &  $-\check{p}_2$ & $+\check{p}_7$ & $-\check{p}_8$ &  $-\check{p}_5$ &  $+\check{p}_6$\\
$+\check{p}_4$ &  $-\check{p}_3$ &  $+\check{p}_2$ &  $-\check{p}_1$ & $-\check{p}_8$ & $-\check{p}_7$ &  $+\check{p}_6$ &  $+\check{p}_5$\\
$+\check{p}_5$ &  $+\check{p}_6$ &  $-\check{p}_7$ &  $+\check{p}_8$ & $-\check{p}_1$ & $-\check{p}_2$ &  $+\check{p}_3$ &  $-\check{p}_4$\\
$+\check{p}_6$ &  $-\check{p}_5$ &  $+\check{p}_8$ &  $+\check{p}_7$ & $+\check{p}_2$ & $-\check{p}_1$ &  $-\check{p}_4$ &  $-\check{p}_3$\\
$+\check{p}_7$ &  $+\check{p}_8$ &  $+\check{p}_5$ &  $-\check{p}_6$ & $-\check{p}_3$ & $+\check{p}_4$ &  $-\check{p}_1$ &  $-\check{p}_2$\\
$+\check{p}_8$ &  $-\check{p}_7$ &  $-\check{p}_6$ &  $-\check{p}_5$ & $+\check{p}_4$ & $+\check{p}_3$ &  $+\check{p}_2$ &  $-\check{p}_1$\\
\end{tabular}
\right)$ &\hspace{-12pt} \tiny $\left(
\begin{tabular}{cccccccc}%
$+\check{p}_1$ &  $+\check{p}_2$ &  $+\check{p}_3$ &  $+\check{p}_4$ & $+\check{p}_5$ & $+\check{p}_6$ &  $+\check{p}_7$ &  $+\check{p}_8$\\
$+\check{p}_2$ &  $-\check{p}_1$ &  $-\check{p}_4$ &  $+\check{p}_3$ & $-\check{p}_6$ & $+\check{p}_5$ &  $-\check{p}_8$ &  $+\check{p}_7$\\
$+\check{p}_3$ &  $+\check{p}_4$ &  $-\check{p}_1$ &  $-\check{p}_2$ & $-\check{p}_7$ & $+\check{p}_8$ &  $+\check{p}_5$ &  $-\check{p}_6$\\
$+\check{p}_4$ &  $-\check{p}_3$ &  $+\check{p}_2$ &  $-\check{p}_1$ & $+\check{p}_8$ & $+\check{p}_7$ &  $-\check{p}_6$ &  $-\check{p}_5$\\
$+\check{p}_5$ &  $+\check{p}_6$ &  $+\check{p}_7$ &  $-\check{p}_8$ & $-\check{p}_1$ & $-\check{p}_2$ &  $-\check{p}_3$ &  $+\check{p}_4$\\
$+\check{p}_6$ &  $-\check{p}_5$ &  $-\check{p}_8$ &  $-\check{p}_7$ & $+\check{p}_2$ & $-\check{p}_1$ &  $+\check{p}_4$ &  $+\check{p}_3$\\
$+\check{p}_7$ &  $+\check{p}_8$ &  $-\check{p}_5$ &  $+\check{p}_6$ & $+\check{p}_3$ & $-\check{p}_4$ &  $-\check{p}_1$ &  $-\check{p}_2$\\
$+\check{p}_8$ &  $-\check{p}_7$ &  $+\check{p}_6$ &  $+\check{p}_5$ & $-\check{p}_4$ & $-\check{p}_3$ &  $+\check{p}_2$ &  $-\check{p}_1$\\
\end{tabular}
\right)$\\\vspace{12pt} \hspace{-14pt}\tiny $\left(
\begin{tabular}{cccccccc}%
$+\check{p}_1$ &  $+\check{p}_2$ &  $+\check{p}_3$ &  $+\check{p}_4$ & $+\check{p}_5$ & $+\check{p}_6$ &  $+\check{p}_7$ &  $+\check{p}_8$\\
$+\check{p}_2$ &  $-\check{p}_1$ &  $-\check{p}_4$ &  $+\check{p}_3$ & $-\check{p}_6$ & $+\check{p}_5$ &  $+\check{p}_8$ &  $-\check{p}_7$\\
$+\check{p}_3$ &  $+\check{p}_4$ &  $-\check{p}_1$ &  $-\check{p}_2$ & $-\check{p}_7$ & $-\check{p}_8$ &  $+\check{p}_5$ &  $+\check{p}_6$\\
$+\check{p}_4$ &  $-\check{p}_3$ &  $+\check{p}_2$ &  $-\check{p}_1$ & $-\check{p}_8$ & $+\check{p}_7$ &  $-\check{p}_6$ &  $+\check{p}_5$\\
$+\check{p}_5$ &  $+\check{p}_6$ &  $+\check{p}_7$ &  $+\check{p}_8$ & $-\check{p}_1$ & $-\check{p}_2$ &  $-\check{p}_3$ &  $-\check{p}_4$\\
$+\check{p}_6$ &  $-\check{p}_5$ &  $+\check{p}_8$ &  $-\check{p}_7$ & $+\check{p}_2$ & $-\check{p}_1$ &  $+\check{p}_4$ &  $-\check{p}_3$\\
$+\check{p}_7$ &  $-\check{p}_8$ &  $-\check{p}_5$ &  $+\check{p}_6$ & $+\check{p}_3$ & $-\check{p}_4$ &  $-\check{p}_1$ &  $+\check{p}_2$\\
$+\check{p}_8$ &  $+\check{p}_7$ &  $-\check{p}_6$ &  $-\check{p}_5$ & $+\check{p}_4$ & $+\check{p}_3$ &  $-\check{p}_2$ &  $-\check{p}_1$\\
\end{tabular}
\right)$ &\hspace{-12pt} \tiny $\left(
\begin{tabular}{cccccccc}%
$+\check{p}_1$ &  $+\check{p}_2$ &  $+\check{p}_3$ &  $+\check{p}_4$ & $+\check{p}_5$ & $+\check{p}_6$ &  $+\check{p}_7$ &  $+\check{p}_8$\\
$+\check{p}_2$ &  $-\check{p}_1$ &  $-\check{p}_4$ &  $+\check{p}_3$ & $-\check{p}_6$ & $+\check{p}_5$ &  $+\check{p}_8$ &  $-\check{p}_7$\\
$+\check{p}_3$ &  $+\check{p}_4$ &  $-\check{p}_1$ &  $-\check{p}_2$ & $+\check{p}_7$ & $+\check{p}_8$ &  $-\check{p}_5$ &  $-\check{p}_6$\\
$+\check{p}_4$ &  $-\check{p}_3$ &  $+\check{p}_2$ &  $-\check{p}_1$ & $+\check{p}_8$ & $-\check{p}_7$ &  $+\check{p}_6$ &  $-\check{p}_5$\\
$+\check{p}_5$ &  $+\check{p}_6$ &  $-\check{p}_7$ &  $-\check{p}_8$ & $-\check{p}_1$ & $-\check{p}_2$ &  $+\check{p}_3$ &  $+\check{p}_4$\\
$+\check{p}_6$ &  $-\check{p}_5$ &  $-\check{p}_8$ &  $+\check{p}_7$ & $+\check{p}_2$ & $-\check{p}_1$ &  $-\check{p}_4$ &  $+\check{p}_3$\\
$+\check{p}_7$ &  $-\check{p}_8$ &  $+\check{p}_5$ &  $-\check{p}_6$ & $-\check{p}_3$ & $+\check{p}_4$ &  $-\check{p}_1$ &  $+\check{p}_2$\\
$+\check{p}_8$ &  $+\check{p}_7$ &  $+\check{p}_6$ &  $+\check{p}_5$ & $-\check{p}_4$ & $-\check{p}_3$ &  $-\check{p}_2$ &  $-\check{p}_1$\\
\end{tabular}
\right)$\\\vspace{12pt} \hspace{-14pt}\tiny $\left(
\begin{tabular}{cccccccc}%
$+\check{p}_1$ &  $+\check{p}_2$ &  $+\check{p}_3$ &  $+\check{p}_4$ & $+\check{p}_5$ & $+\check{p}_6$ &  $+\check{p}_7$ &  $+\check{p}_8$\\
$+\check{p}_2$ &  $-\check{p}_1$ &  $-\check{p}_4$ &  $+\check{p}_3$ & $+\check{p}_6$ & $-\check{p}_5$ &  $+\check{p}_8$ &  $-\check{p}_7$\\
$+\check{p}_3$ &  $+\check{p}_4$ &  $-\check{p}_1$ &  $-\check{p}_2$ & $-\check{p}_7$ & $+\check{p}_8$ &  $+\check{p}_5$ &  $-\check{p}_6$\\
$+\check{p}_4$ &  $-\check{p}_3$ &  $+\check{p}_2$ &  $-\check{p}_1$ & $-\check{p}_8$ & $-\check{p}_7$ &  $+\check{p}_6$ &  $+\check{p}_5$\\
$+\check{p}_5$ &  $-\check{p}_6$ &  $+\check{p}_7$ &  $+\check{p}_8$ & $-\check{p}_1$ & $+\check{p}_2$ &  $-\check{p}_3$ &  $-\check{p}_4$\\
$+\check{p}_6$ &  $+\check{p}_5$ &  $-\check{p}_8$ &  $+\check{p}_7$ & $-\check{p}_2$ & $-\check{p}_1$ &  $-\check{p}_4$ &  $+\check{p}_3$\\
$+\check{p}_7$ &  $-\check{p}_8$ &  $-\check{p}_5$ &  $-\check{p}_6$ & $+\check{p}_3$ & $+\check{p}_4$ &  $-\check{p}_1$ &  $+\check{p}_2$\\
$+\check{p}_8$ &  $+\check{p}_7$ &  $+\check{p}_6$ &  $-\check{p}_5$ & $+\check{p}_4$ & $-\check{p}_3$ &  $-\check{p}_2$ &  $-\check{p}_1$\\
\end{tabular}
\right)$ &\hspace{-12pt} \tiny $\left(
\begin{tabular}{cccccccc}%
$+\check{p}_1$ &  $+\check{p}_2$ &  $+\check{p}_3$ &  $+\check{p}_4$ & $+\check{p}_5$ & $+\check{p}_6$ &  $+\check{p}_7$ &  $+\check{p}_8$\\
$+\check{p}_2$ &  $-\check{p}_1$ &  $-\check{p}_4$ &  $+\check{p}_3$ & $+\check{p}_6$ & $-\check{p}_5$ &  $-\check{p}_8$ &  $+\check{p}_7$\\
$+\check{p}_3$ &  $+\check{p}_4$ &  $-\check{p}_1$ &  $-\check{p}_2$ & $-\check{p}_7$ & $-\check{p}_8$ &  $+\check{p}_5$ &  $+\check{p}_6$\\
$+\check{p}_4$ &  $-\check{p}_3$ &  $+\check{p}_2$ &  $-\check{p}_1$ & $+\check{p}_8$ & $-\check{p}_7$ &  $+\check{p}_6$ &  $-\check{p}_5$\\
$+\check{p}_5$ &  $-\check{p}_6$ &  $+\check{p}_7$ &  $-\check{p}_8$ & $-\check{p}_1$ & $+\check{p}_2$ &  $-\check{p}_3$ &  $+\check{p}_4$\\
$+\check{p}_6$ &  $+\check{p}_5$ &  $+\check{p}_8$ &  $+\check{p}_7$ & $-\check{p}_2$ & $-\check{p}_1$ &  $-\check{p}_4$ &  $-\check{p}_3$\\
$+\check{p}_7$ &  $+\check{p}_8$ &  $-\check{p}_5$ &  $-\check{p}_6$ & $+\check{p}_3$ & $+\check{p}_4$ &  $-\check{p}_1$ &  $-\check{p}_2$\\
$+\check{p}_8$ &  $-\check{p}_7$ &  $-\check{p}_6$ &  $+\check{p}_5$ & $-\check{p}_4$ & $+\check{p}_3$ &  $+\check{p}_2$ &  $-\check{p}_1$\\
\end{tabular}
\right)$\\\vspace{12pt} \hspace{-14pt}\tiny $\left(
\begin{tabular}{cccccccc}%
$+\check{p}_1$ &  $+\check{p}_2$ &  $+\check{p}_3$ &  $+\check{p}_4$ & $+\check{p}_5$ & $+\check{p}_6$ &  $+\check{p}_7$ &  $+\check{p}_8$\\
$+\check{p}_2$ &  $-\check{p}_1$ &  $-\check{p}_4$ &  $+\check{p}_3$ & $+\check{p}_6$ & $-\check{p}_5$ &  $-\check{p}_8$ &  $+\check{p}_7$\\
$+\check{p}_3$ &  $+\check{p}_4$ &  $-\check{p}_1$ &  $-\check{p}_2$ & $+\check{p}_7$ & $+\check{p}_8$ &  $-\check{p}_5$ &  $-\check{p}_6$\\
$+\check{p}_4$ &  $-\check{p}_3$ &  $+\check{p}_2$ &  $-\check{p}_1$ & $-\check{p}_8$ & $+\check{p}_7$ &  $-\check{p}_6$ &  $+\check{p}_5$\\
$+\check{p}_5$ &  $-\check{p}_6$ &  $-\check{p}_7$ &  $+\check{p}_8$ & $-\check{p}_1$ & $+\check{p}_2$ &  $+\check{p}_3$ &  $-\check{p}_4$\\
$+\check{p}_6$ &  $+\check{p}_5$ &  $-\check{p}_8$ &  $-\check{p}_7$ & $-\check{p}_2$ & $-\check{p}_1$ &  $+\check{p}_4$ &  $+\check{p}_3$\\
$+\check{p}_7$ &  $+\check{p}_8$ &  $+\check{p}_5$ &  $+\check{p}_6$ & $-\check{p}_3$ & $-\check{p}_4$ &  $-\check{p}_1$ &  $-\check{p}_2$\\
$+\check{p}_8$ &  $-\check{p}_7$ &  $+\check{p}_6$ &  $-\check{p}_5$ & $+\check{p}_4$ & $-\check{p}_3$ &  $+\check{p}_2$ &  $-\check{p}_1$\\
\end{tabular}
\right)$ &\hspace{-12pt} \tiny $\left(
\begin{tabular}{cccccccc}%
$+\check{p}_1$ &  $+\check{p}_2$ &  $+\check{p}_3$ &  $+\check{p}_4$ & $+\check{p}_5$ & $+\check{p}_6$ &  $+\check{p}_7$ &  $+\check{p}_8$\\
$+\check{p}_2$ &  $-\check{p}_1$ &  $-\check{p}_4$ &  $+\check{p}_3$ & $+\check{p}_6$ & $-\check{p}_5$ &  $+\check{p}_8$ &  $-\check{p}_7$\\
$+\check{p}_3$ &  $+\check{p}_4$ &  $-\check{p}_1$ &  $-\check{p}_2$ & $+\check{p}_7$ & $-\check{p}_8$ &  $-\check{p}_5$ &  $+\check{p}_6$\\
$+\check{p}_4$ &  $-\check{p}_3$ &  $+\check{p}_2$ &  $-\check{p}_1$ & $+\check{p}_8$ & $+\check{p}_7$ &  $-\check{p}_6$ &  $-\check{p}_5$\\
$+\check{p}_5$ &  $-\check{p}_6$ &  $-\check{p}_7$ &  $-\check{p}_8$ & $-\check{p}_1$ & $+\check{p}_2$ &  $+\check{p}_3$ &  $+\check{p}_4$\\
$+\check{p}_6$ &  $+\check{p}_5$ &  $+\check{p}_8$ &  $-\check{p}_7$ & $-\check{p}_2$ & $-\check{p}_1$ &  $+\check{p}_4$ &  $-\check{p}_3$\\
$+\check{p}_7$ &  $-\check{p}_8$ &  $+\check{p}_5$ &  $+\check{p}_6$ & $-\check{p}_3$ & $-\check{p}_4$ &  $-\check{p}_1$ &  $+\check{p}_2$\\
$+\check{p}_8$ &  $+\check{p}_7$ &  $-\check{p}_6$ &  $+\check{p}_5$ & $-\check{p}_4$ & $+\check{p}_3$ &  $-\check{p}_2$ &  $-\check{p}_1$\\
\end{tabular}
\right)$\\\vspace{12pt} \hspace{-14pt}\tiny $\left(
\begin{tabular}{cccccccc}%
$+\check{p}_1$ &  $+\check{p}_2$ &  $+\check{p}_3$ &  $+\check{p}_4$ & $+\check{p}_5$ & $+\check{p}_6$ &  $+\check{p}_7$ &  $+\check{p}_8$\\
$+\check{p}_2$ &  $-\check{p}_1$ &  $+\check{p}_4$ &  $-\check{p}_3$ & $-\check{p}_6$ & $+\check{p}_5$ &  $-\check{p}_8$ &  $+\check{p}_7$\\
$+\check{p}_3$ &  $-\check{p}_4$ &  $-\check{p}_1$ &  $+\check{p}_2$ & $-\check{p}_7$ & $+\check{p}_8$ &  $+\check{p}_5$ &  $-\check{p}_6$\\
$+\check{p}_4$ &  $+\check{p}_3$ &  $-\check{p}_2$ &  $-\check{p}_1$ & $-\check{p}_8$ & $-\check{p}_7$ &  $+\check{p}_6$ &  $+\check{p}_5$\\
$+\check{p}_5$ &  $+\check{p}_6$ &  $+\check{p}_7$ &  $+\check{p}_8$ & $-\check{p}_1$ & $-\check{p}_2$ &  $-\check{p}_3$ &  $-\check{p}_4$\\
$+\check{p}_6$ &  $-\check{p}_5$ &  $-\check{p}_8$ &  $+\check{p}_7$ & $+\check{p}_2$ & $-\check{p}_1$ &  $-\check{p}_4$ &  $+\check{p}_3$\\
$+\check{p}_7$ &  $+\check{p}_8$ &  $-\check{p}_5$ &  $-\check{p}_6$ & $+\check{p}_3$ & $+\check{p}_4$ &  $-\check{p}_1$ &  $-\check{p}_2$\\
$+\check{p}_8$ &  $-\check{p}_7$ &  $+\check{p}_6$ &  $-\check{p}_5$ & $+\check{p}_4$ & $-\check{p}_3$ &  $+\check{p}_2$ &  $-\check{p}_1$\\
\end{tabular}
\right)$ &\hspace{-12pt} \tiny $\left(
\begin{tabular}{cccccccc}%
$+\check{p}_1$ &  $+\check{p}_2$ &  $+\check{p}_3$ &  $+\check{p}_4$ & $+\check{p}_5$ & $+\check{p}_6$ &  $+\check{p}_7$ &  $+\check{p}_8$\\
$+\check{p}_2$ &  $-\check{p}_1$ &  $+\check{p}_4$ &  $-\check{p}_3$ & $-\check{p}_6$ & $+\check{p}_5$ &  $+\check{p}_8$ &  $-\check{p}_7$\\
$+\check{p}_3$ &  $-\check{p}_4$ &  $-\check{p}_1$ &  $+\check{p}_2$ & $-\check{p}_7$ & $-\check{p}_8$ &  $+\check{p}_5$ &  $+\check{p}_6$\\
$+\check{p}_4$ &  $+\check{p}_3$ &  $-\check{p}_2$ &  $-\check{p}_1$ & $+\check{p}_8$ & $-\check{p}_7$ &  $+\check{p}_6$ &  $-\check{p}_5$\\
$+\check{p}_5$ &  $+\check{p}_6$ &  $+\check{p}_7$ &  $-\check{p}_8$ & $-\check{p}_1$ & $-\check{p}_2$ &  $-\check{p}_3$ &  $+\check{p}_4$\\
$+\check{p}_6$ &  $-\check{p}_5$ &  $+\check{p}_8$ &  $+\check{p}_7$ & $+\check{p}_2$ & $-\check{p}_1$ &  $-\check{p}_4$ &  $-\check{p}_3$\\
$+\check{p}_7$ &  $-\check{p}_8$ &  $-\check{p}_5$ &  $-\check{p}_6$ & $+\check{p}_3$ & $+\check{p}_4$ &  $-\check{p}_1$ &  $+\check{p}_2$\\
$+\check{p}_8$ &  $+\check{p}_7$ &  $-\check{p}_6$ &  $+\check{p}_5$ & $-\check{p}_4$ & $+\check{p}_3$ &  $-\check{p}_2$ &  $-\check{p}_1$\\
\end{tabular}
\right)$\\\vspace{12pt} \hspace{-14pt}\tiny $\left(
\begin{tabular}{cccccccc}%
$+\check{p}_1$ &  $+\check{p}_2$ &  $+\check{p}_3$ &  $+\check{p}_4$ & $+\check{p}_5$ & $+\check{p}_6$ &  $+\check{p}_7$ &  $+\check{p}_8$\\
$+\check{p}_2$ &  $-\check{p}_1$ &  $+\check{p}_4$ &  $-\check{p}_3$ & $-\check{p}_6$ & $+\check{p}_5$ &  $+\check{p}_8$ &  $-\check{p}_7$\\
$+\check{p}_3$ &  $-\check{p}_4$ &  $-\check{p}_1$ &  $+\check{p}_2$ & $+\check{p}_7$ & $+\check{p}_8$ &  $-\check{p}_5$ &  $-\check{p}_6$\\
$+\check{p}_4$ &  $+\check{p}_3$ &  $-\check{p}_2$ &  $-\check{p}_1$ & $-\check{p}_8$ & $+\check{p}_7$ &  $-\check{p}_6$ &  $+\check{p}_5$\\
$+\check{p}_5$ &  $+\check{p}_6$ &  $-\check{p}_7$ &  $+\check{p}_8$ & $-\check{p}_1$ & $-\check{p}_2$ &  $+\check{p}_3$ &  $-\check{p}_4$\\
$+\check{p}_6$ &  $-\check{p}_5$ &  $-\check{p}_8$ &  $-\check{p}_7$ & $+\check{p}_2$ & $-\check{p}_1$ &  $+\check{p}_4$ &  $+\check{p}_3$\\
$+\check{p}_7$ &  $-\check{p}_8$ &  $+\check{p}_5$ &  $+\check{p}_6$ & $-\check{p}_3$ & $-\check{p}_4$ &  $-\check{p}_1$ &  $+\check{p}_2$\\
$+\check{p}_8$ &  $+\check{p}_7$ &  $+\check{p}_6$ &  $-\check{p}_5$ & $+\check{p}_4$ & $-\check{p}_3$ &  $-\check{p}_2$ &  $-\check{p}_1$\\
\end{tabular}
\right)$ &\hspace{-12pt} \tiny $\left(
\begin{tabular}{cccccccc}%
$+\check{p}_1$ &  $+\check{p}_2$ &  $+\check{p}_3$ &  $+\check{p}_4$ & $+\check{p}_5$ & $+\check{p}_6$ &  $+\check{p}_7$ &  $+\check{p}_8$\\
$+\check{p}_2$ &  $-\check{p}_1$ &  $+\check{p}_4$ &  $-\check{p}_3$ & $-\check{p}_6$ & $+\check{p}_5$ &  $-\check{p}_8$ &  $+\check{p}_7$\\
$+\check{p}_3$ &  $-\check{p}_4$ &  $-\check{p}_1$ &  $+\check{p}_2$ & $+\check{p}_7$ & $-\check{p}_8$ &  $-\check{p}_5$ &  $+\check{p}_6$\\
$+\check{p}_4$ &  $+\check{p}_3$ &  $-\check{p}_2$ &  $-\check{p}_1$ & $+\check{p}_8$ & $+\check{p}_7$ &  $-\check{p}_6$ &  $-\check{p}_5$\\
$+\check{p}_5$ &  $+\check{p}_6$ &  $-\check{p}_7$ &  $-\check{p}_8$ & $-\check{p}_1$ & $-\check{p}_2$ &  $+\check{p}_3$ &  $+\check{p}_4$\\
$+\check{p}_6$ &  $-\check{p}_5$ &  $+\check{p}_8$ &  $-\check{p}_7$ & $+\check{p}_2$ & $-\check{p}_1$ &  $+\check{p}_4$ &  $-\check{p}_3$\\
$+\check{p}_7$ &  $+\check{p}_8$ &  $+\check{p}_5$ &  $+\check{p}_6$ & $-\check{p}_3$ & $-\check{p}_4$ &  $-\check{p}_1$ &  $-\check{p}_2$\\
$+\check{p}_8$ &  $-\check{p}_7$ &  $-\check{p}_6$ &  $+\check{p}_5$ & $-\check{p}_4$ & $+\check{p}_3$ &  $+\check{p}_2$ &  $-\check{p}_1$\\
\end{tabular}
\right)$\\\vspace{12pt} \hspace{-14pt}\tiny $\left(
\begin{tabular}{cccccccc}%
$+\check{p}_1$ &  $+\check{p}_2$ &  $+\check{p}_3$ &  $+\check{p}_4$ & $+\check{p}_5$ & $+\check{p}_6$ &  $+\check{p}_7$ &  $+\check{p}_8$\\
$+\check{p}_2$ &  $-\check{p}_1$ &  $+\check{p}_4$ &  $-\check{p}_3$ & $+\check{p}_6$ & $-\check{p}_5$ &  $-\check{p}_8$ &  $+\check{p}_7$\\
$+\check{p}_3$ &  $-\check{p}_4$ &  $-\check{p}_1$ &  $+\check{p}_2$ & $-\check{p}_7$ & $-\check{p}_8$ &  $+\check{p}_5$ &  $+\check{p}_6$\\
$+\check{p}_4$ &  $+\check{p}_3$ &  $-\check{p}_2$ &  $-\check{p}_1$ & $-\check{p}_8$ & $+\check{p}_7$ &  $-\check{p}_6$ &  $+\check{p}_5$\\
$+\check{p}_5$ &  $-\check{p}_6$ &  $+\check{p}_7$ &  $+\check{p}_8$ & $-\check{p}_1$ & $+\check{p}_2$ &  $-\check{p}_3$ &  $-\check{p}_4$\\
$+\check{p}_6$ &  $+\check{p}_5$ &  $+\check{p}_8$ &  $-\check{p}_7$ & $-\check{p}_2$ & $-\check{p}_1$ &  $+\check{p}_4$ &  $-\check{p}_3$\\
$+\check{p}_7$ &  $+\check{p}_8$ &  $-\check{p}_5$ &  $+\check{p}_6$ & $+\check{p}_3$ & $-\check{p}_4$ &  $-\check{p}_1$ &  $-\check{p}_2$\\
$+\check{p}_8$ &  $-\check{p}_7$ &  $-\check{p}_6$ &  $-\check{p}_5$ & $+\check{p}_4$ & $+\check{p}_3$ &  $+\check{p}_2$ &  $-\check{p}_1$\\
\end{tabular}
\right)$ &\hspace{-12pt} \tiny $\left(
\begin{tabular}{cccccccc}%
$+\check{p}_1$ &  $+\check{p}_2$ &  $+\check{p}_3$ &  $+\check{p}_4$ & $+\check{p}_5$ & $+\check{p}_6$ &  $+\check{p}_7$ &  $+\check{p}_8$\\
$+\check{p}_2$ &  $-\check{p}_1$ &  $+\check{p}_4$ &  $-\check{p}_3$ & $+\check{p}_6$ & $-\check{p}_5$ &  $+\check{p}_8$ &  $-\check{p}_7$\\
$+\check{p}_3$ &  $-\check{p}_4$ &  $-\check{p}_1$ &  $+\check{p}_2$ & $-\check{p}_7$ & $+\check{p}_8$ &  $+\check{p}_5$ &  $-\check{p}_6$\\
$+\check{p}_4$ &  $+\check{p}_3$ &  $-\check{p}_2$ &  $-\check{p}_1$ & $+\check{p}_8$ & $+\check{p}_7$ &  $-\check{p}_6$ &  $-\check{p}_5$\\
$+\check{p}_5$ &  $-\check{p}_6$ &  $+\check{p}_7$ &  $-\check{p}_8$ & $-\check{p}_1$ & $+\check{p}_2$ &  $-\check{p}_3$ &  $+\check{p}_4$\\
$+\check{p}_6$ &  $+\check{p}_5$ &  $-\check{p}_8$ &  $-\check{p}_7$ & $-\check{p}_2$ & $-\check{p}_1$ &  $+\check{p}_4$ &  $+\check{p}_3$\\
$+\check{p}_7$ &  $-\check{p}_8$ &  $-\check{p}_5$ &  $+\check{p}_6$ & $+\check{p}_3$ & $-\check{p}_4$ &  $-\check{p}_1$ &  $+\check{p}_2$\\
$+\check{p}_8$ &  $+\check{p}_7$ &  $+\check{p}_6$ &  $+\check{p}_5$ & $-\check{p}_4$ & $-\check{p}_3$ &  $-\check{p}_2$ &  $-\check{p}_1$\\
\end{tabular}
\right)$\\\vspace{12pt} \hspace{-14pt}\tiny $\left(
\begin{tabular}{cccccccc}
$+\check{p}_1$ &  $+\check{p}_2$ &  $+\check{p}_3$ &  $+\check{p}_4$ & $+\check{p}_5$ & $+\check{p}_6$ &  $+\check{p}_7$ &  $+\check{p}_8$\\
$+\check{p}_2$ &  $-\check{p}_1$ &  $+\check{p}_4$ &  $-\check{p}_3$ & $+\check{p}_6$ & $-\check{p}_5$ &  $+\check{p}_8$ &  $-\check{p}_7$\\
$+\check{p}_3$ &  $-\check{p}_4$ &  $-\check{p}_1$ &  $+\check{p}_2$ & $+\check{p}_7$ & $-\check{p}_8$ &  $-\check{p}_5$ &  $+\check{p}_6$\\
$+\check{p}_4$ &  $+\check{p}_3$ &  $-\check{p}_2$ &  $-\check{p}_1$ & $-\check{p}_8$ & $-\check{p}_7$ &  $+\check{p}_6$ &  $+\check{p}_5$\\
$+\check{p}_5$ &  $-\check{p}_6$ &  $-\check{p}_7$ &  $+\check{p}_8$ & $-\check{p}_1$ & $+\check{p}_2$ &  $+\check{p}_3$ &  $-\check{p}_4$\\
$+\check{p}_6$ &  $+\check{p}_5$ &  $+\check{p}_8$ &  $+\check{p}_7$ & $-\check{p}_2$ & $-\check{p}_1$ &  $-\check{p}_4$ &  $-\check{p}_3$\\
$+\check{p}_7$ &  $-\check{p}_8$ &  $+\check{p}_5$ &  $-\check{p}_6$ & $-\check{p}_3$ & $+\check{p}_4$ &  $-\check{p}_1$ &  $+\check{p}_2$\\
$+\check{p}_8$ &  $+\check{p}_7$ &  $-\check{p}_6$ &  $-\check{p}_5$ & $+\check{p}_4$ & $+\check{p}_3$ &  $-\check{p}_2$ &  $-\check{p}_1$\\
\end{tabular}
\right)$ &\hspace{-12pt} \tiny $\left(
\begin{tabular}{cccccccc}
$+\check{p}_1$ &  $+\check{p}_2$ &  $+\check{p}_3$ &  $+\check{p}_4$ & $+\check{p}_5$ & $+\check{p}_6$ &  $+\check{p}_7$ &  $+\check{p}_8$\\
$+\check{p}_2$ &  $-\check{p}_1$ &  $+\check{p}_4$ &  $-\check{p}_3$ & $+\check{p}_6$ & $-\check{p}_5$ &  $-\check{p}_8$ &  $+\check{p}_7$\\
$+\check{p}_3$ &  $-\check{p}_4$ &  $-\check{p}_1$ &  $+\check{p}_2$ & $+\check{p}_7$ & $+\check{p}_8$ &  $-\check{p}_5$ &  $-\check{p}_6$\\
$+\check{p}_4$ &  $+\check{p}_3$ &  $-\check{p}_2$ &  $-\check{p}_1$ & $+\check{p}_8$ & $-\check{p}_7$ &  $+\check{p}_6$ &  $-\check{p}_5$\\
$+\check{p}_5$ &  $-\check{p}_6$ &  $-\check{p}_7$ &  $-\check{p}_8$ & $-\check{p}_1$ & $+\check{p}_2$ &  $+\check{p}_3$ &  $+\check{p}_4$\\
$+\check{p}_6$ &  $+\check{p}_5$ &  $-\check{p}_8$ &  $+\check{p}_7$ & $-\check{p}_2$ & $-\check{p}_1$ &  $-\check{p}_4$ &  $+\check{p}_3$\\
$+\check{p}_7$ &  $+\check{p}_8$ &  $+\check{p}_5$ &  $-\check{p}_6$ & $-\check{p}_3$ & $+\check{p}_4$ &  $-\check{p}_1$ &  $-\check{p}_2$\\
$+\check{p}_8$ &  $-\check{p}_7$ &  $+\check{p}_6$ &  $+\check{p}_5$ & $-\check{p}_4$ & $-\check{p}_3$ &  $+\check{p}_2$ &  $-\check{p}_1$\\
\end{tabular}
\right)$\\
\end{tabular}
\caption{All possible orthogonal $4\times4$ and $8\times8$
Latin-Hadamard matrices. Here $\check{p}_i$ denotes $\sqrt{p_i}$,
for typographical reasons.}\label{T: LatHad4AND8}
\end{table}

\section{Connections with Algebra}\label{S:DivisionAlgebra}

In the first subsection we give an explanation to the non-existence
of valid colorings for size $16$ and higher. We then show, in
Subsection~\ref{SS:tradeoff} how to remedy this issue by taking some
of probability values to be equal, using the theory of orthogonal
designs.

\subsection{Division Algebras Connections}
An algebra over a field is a vector space that is equipped with a
bilinear multiplication operator. Three well-known algebras over the
field of real numbers are the two dimensional algebra of complex
numbers, the four dimensional algebra of quaternions, and the eight
dimensional non-associative algebra of octonions.The octonions are
constructed from the quaternions in a way similar to the
construction of the quaternions from the complex numbers, and the
construction of the latter from the real numbers. This sort of
construction can be repeated indefinitely via a classical doubling
procedure called the Cayley-Dickson construction. The next algebra
produced by doubling the octonions is the less known algebra of
sedenions, which has been mostly neglected by mathematicians,
perhaps because it is not a division algebra. That is, in this
algebra the equations $ax=b$ and $xa=b$ need not have a solution in
$x$. Equivalently, there exist zero divisors in the sedenions
algebra, i.e., there are elements $a\neq0$ and $b\neq0$ but $ab=0$.
This is not possible in the first three algebras above.

Usually the multiplication operation in an algebra is defined via a
multiplication table for elements of a basis. As an illustration,
the basis elements of quaternions are denoted by $\textbf{1}$,
$\textbf{i}$, $\textbf{j}$, $\textbf{k}$. Multiplication of
quaternions is defined by
$\textbf{i}^2=\textbf{j}^2=\textbf{k}^2=-1$,
$\textbf{ij}=\textbf{k}$, $\textbf{jk}=\textbf{i}$,
$\textbf{ki}=\textbf{j}$, $\textbf{ji}=-\textbf{k}$,
$\textbf{kj}=-\textbf{i}$, $\textbf{ik}=-\textbf{j}$. Denoting
$\textbf{1}$, $\textbf{i}$, $\textbf{j}$ and $\textbf{k}$ by $e_1$,
$e_2$, $e_3$ and $e_4$ respectively, the multiplication table can be
seen as a $4\times4$ matrix.
\begin{figure}[h!]
\center $\left(
\begin{tabular}{c|cccc}
      & $e_1$ & $e_2$ & $e_3$ & $e_4$\\
\hline
$e_1$ & $+e_1$ & $+e_2$ & $+e_3$ & $+e_4$\\
$e_2$ & $+e_2$ & $-e_1$ & $+e_4$ & $-e_3$\\
$e_3$ & $+e_3$ & $-e_4$ & $-e_1$ & $+e_2$\\
$e_4$ & $+e_4$ & $+e_3$ & $-e_2$ & $-e_1$
\end{tabular}
\right)$
\end{figure}

Note that the resulting $4\times4$ matrix consists of an orthogonal
matrix with the AB-BA property. A similar phenomenon is noticed if
we write the multiplication table of octonions. The multiplication
table of the sedenions also has the AB-BA property but the columns
are no more orthogonal.

A classical theorem in algebra, Hurwitz Theorem, states that the
only division algebras over the field of real numbers are the
algebras of real numbers, complex numbers, quaternions and
octonions, see \cite{Jacobson}, \cite{HypercomplexNumbers} for
example. The connection to our problem is established by seeing that
the existence of zero divisors is associated with the existence of
two columns that are not orthogonal. In effect, let $e_1,\cdots,
e_{16}$ be a basis of a $16$-dimensional algebra generated by a
coloring of our matrix. Suppose also that $e_1$ is the unit element
while the rest are the imaginary roots of $-1$. de Marrais
\cite{deMarrais} shows that the zero divisors are only of the form
$(e_i\pm e_j)$ for some $i,j\neq1$, see also
\cite{GuillermoMoreno97}. Consider some basis elements $k_1, k_2,
l_1,l_2$ all different from unity and suppose that
$k_1l_1=\alpha_{11}e_a$, $k_1l_2=\alpha_{12}e_b$,
$k_2l_1=\alpha_{21}e_{b^{\prime}}$,
$k_2l_2=\alpha_{22}e_{a^{\prime}}$, where the $\alpha's$ are $\pm1$.
Since the multiplication table without the signs makes a Latin
square, note that $a\neq b$, $a\neq b^{\prime}$, $a^{\prime}\neq b$
and $a^{\prime}\neq b^{\prime}$. If $a\neq a^{\prime}$ or $b\neq
b^{\prime}$ then one can readily check that none of the four
products $(k_1\pm k_2)(l_1\pm l_2)$ can be zero. For $a=a^{\prime}$
and $b=b^{\prime}$, $(k_1-k_2)(l_1+l_2)=0$ if and only if
$\alpha_{11}=\alpha_{22}$ and $\alpha_{12}=\alpha_{21}$. Likewise,
$(k_1+k_2)(l_1+l_2)=0$ if and only if $\alpha_{11}=-\alpha_{22}$ and
$\alpha_{12}=-\alpha_{21}$. In both scenarios, and using our
terminology, we do not obtain a `correct' coloring of the AB-BA
corners. In other words, the existence of a (complete) coloring of
the multiplication table is equivalent to the non-existence of zero
divisors. Hence, the non-existence of a Latin-Hadamard matrix of
order 16 or above rests on the deeper truth stated by the Hurwitz
theorem. The next section shows that this issue can be remedied.


\subsection{Trade-Offs Via Orthogonal Designs}\label{SS:tradeoff}

We have seen that a Latin-Hadamard matrix allows for component tests
that compare pairs of sample proportions in various ways. When the
number of bins is $16$ or a higher power of two, we are faced with
the impossibility of having division algebras over the real numbers.
However, this can be overcome if we are willing to assume some prior
relations between the bin probabilities $p_1,\ldots,p_{16}$. An
obvious example that we already discussed is when we assume that all
of them are equal, in which case Hadamard matrices clearly exist for
all pure powers of two.

It turns out that we have discussed the two extreme cases: when all
cell probabilities are equal and when all are different. The latter
case is thought to occur when no a priori relations are assumed
between the cell probabilities. In both cases we have insisted that
each component test uses all cell counts, by requiring that the
matrix has no zero entries. The mathematical apparatus to study
cases other the two mentioned above is the theory of orthogonal
designs which we will briefly define next.

\begin{definition}
Given any integer $n$, we can write $n=2^a\cdot b$ for some integers
$a$ and $b$ where $b$ is odd. Writing $a=4c+d$; $0\leq d<4$, Radon's
function is defined as $\rho(n)=8c+2^d$.
\end{definition}

It is worth noting that $\rho(n)=n$ if $n=1,2,3,4$. Also
$\rho(16)=9$, $\rho(32)=10$, $\rho(64)=12$, etc.

\begin{definition}
An orthogonal design of order $n$ and type $(s_1,\ldots, s_l)$
($s_i>0$ and $0< l\leq n$ is an integer) on the commuting variables
$x_1,\ldots,x_l$ is an $n\times n$ matrix $A$ with entries from
$\{0,\pm x_1,\ldots,\pm x_l\}$ such that
\[ AA^t=\left(\sum_{i=1}^ls_ix_i^2\right)I_n
\]
\end{definition}
That is, the columns of $A$ are orthogonal and each column has
exactly $s_i$ entries of type $\pm x_i$. Radon's Theorem states the
following.\\

\noindent \textbf{Radon's Theorem}. Let $A$ be an orthogonal design
of order $n$ and type $(s_1,\ldots, s_l)$ on the variables
$x_1,\ldots,x_l$. Then $\rho(n)\geq l$.

It follows by this theorem, and the remark that $\rho(n)=n$ if and
only if $n=1,2,4,8$, that we can only find orthogonal designs of
order $n$ and $n$ variables when $n$ takes one of the four values,
1, 2, 4, 8.

For  a chi-square decomposition with relations between class
probabilities that are less restrictive than those of
Section~\ref{S:prelims}, we make use of orthogonal designs with
nonzero entries. I. Kotsireas and C. Koukouvinos \cite{KK2006}
obtain an orthogonal design of order 16 using the algebra of
sedenions. They start with $16$ indeterminate variables
$A,B,\ldots,P$, and imitate a construction by Becker and
Weispfenning \cite{Becker93} of Hadamard matrices using quaternions.
As a result of this direct imitation, they obtain a $16\times16$
matrix whose entries are polynomial functions in the indeterminate
variables and with identical diagonal entries. To obtain the
orthogonal design they set all non-diagonal entries to zero and thus
obtain a system of 42 equations. Applying a Gr\"{o}bner basis
technique, they obtain an equivalent but reduced set of equations.
Each solution of this reduced system yields a $16\times16$
orthogonal design. They show computationally that all solutions boil
down an at-most-9-variable orthogonal design. They report an
explicit orthogonal design matrix, which we show with
$\sqrt{p_1},\ldots,\sqrt{p_{9}}$ as the indeterminate variables
(with $\check{p}_i$ denoting $\sqrt{p}_i$).

\begin{figure}[h!]
\center \tiny $\left(
\begin{tabular}{cccccccccccccccc}
$\check{p}_1$ & $\check{p}_2$ & $\check{p}_3$ & $\check{p}_4$ & $\check{p}_5$ & $\check{p}_6$ & $\check{p}_7$ & $\check{p}_8$ & $\check{p}_9$ & $\check{p}_2$ & $\check{p}_3$ & $\check{p}_4$ & $\check{p}_5$ & $\check{p}_6$ & $\check{p}_7$ & $\check{p}_8$\\
$-\check{p}_2$ & $\check{p}_1$ & $-\check{p}_4$ & $\check{p}_3$ & $-\check{p}_6$ & $\check{p}_5$ & $\check{p}_8$ & $-\check{p}_7$ & $-\check{p}_2$ & $\check{p}_9$ & $\check{p}_4$ & $-\check{p}_3$ & $\check{p}_6$ & $-\check{p}_5$ & $-\check{p}_8$ & $\check{p}_7$\\
$-\check{p}_3$ & $\check{p}_4$ & $\check{p}_1$ & $-\check{p}_2$ & $-\check{p}_7$ & $-\check{p}_8$ & $\check{p}_5$ & $\check{p}_6$ & $-\check{p}_3$ & $-\check{p}_4$ & $\check{p}_9$ & $\check{p}_2$ & $\check{p}_7$ & $\check{p}_8$ & $-\check{p}_5$ & $-\check{p}_6$\\
$-\check{p}_4$ & $-\check{p}_3$ & $\check{p}_2$ & $\check{p}_1$ & $-\check{p}_8$ & $\check{p}_7$ & $-\check{p}_6$ & $\check{p}_5$ & $-\check{p}_4$ & $\check{p}_3$ & $-\check{p}_2$ & $\check{p}_9$ & $\check{p}_8$ & $-\check{p}_7$& $\check{p}_6$ & $-\check{p}_5$\\
$-\check{p}_5$ & $\check{p}_6$ & $\check{p}_7$ & $\check{p}_8$ & $\check{p}_1$ & $-\check{p}_2$ & $-\check{p}_3$ & $-\check{p}_4$ & $-\check{p}_5$ & $-\check{p}_6$ & $-\check{p}_7$ & $-\check{p}_8$ & $\check{p}_9$ & $\check{p}_2$ & $\check{p}_3$ & $\check{p}_4$\\
$-\check{p}_6$ & $-\check{p}_5$ & $\check{p}_8$ & $-\check{p}_7$ & $\check{p}_2$ & $\check{p}_1$ & $\check{p}_4$ & $-\check{p}_3$ & $-\check{p}_6$ & $\check{p}_5$ & $-\check{p}_8$ & $\check{p}_7$ & $\check{p}_2$ & $-\check{p}_9$ & $\check{p}_4$ & $-\check{p}_3$\\
$-\check{p}_7$ & $-\check{p}_8$ & $-\check{p}_5$ & $\check{p}_6$ & $\check{p}_3$ & $-\check{p}_4$ & $\check{p}_1$ & $\check{p}_2$ & $-\check{p}_7$ & $\check{p}_8$ & $\check{p}_5$ & $-\check{p}_6$ & $-\check{p}_3$ & $\check{p}_4$& $\check{p}_9$ & $-\check{p}_2$\\
$-\check{p}_8$ & $\check{p}_7$ & $-\check{p}_6$ & $-\check{p}_5$ & $\check{p}_4$ & $\check{p}_3$ & $-\check{p}_2$ & $\check{p}_1$ & $-\check{p}_8$ & $-\check{p}_7$ & $\check{p}_6$ & $\check{p}_5$ & $-\check{p}_4$ & $-\check{p}_3$ & $\check{p}_2$ & $\check{p}_9$\\
$-\check{p}_9$ & $\check{p}_2$ & $\check{p}_3$ & $\check{p}_4$ & $\check{p}_5$ & $\check{p}_6$ & $\check{p}_7$ & $\check{p}_8$ & $\check{p}_1$ & $-\check{p}_2$ & $-\check{p}_3$ & $-\check{p}_4$ & $-\check{p}_5$ & $-\check{p}_6$ & $-\check{p}_7$ & $-\check{p}_8$\\
$-\check{p}_2$ & $-\check{p}_9$ & $\check{p}_4$ & $-\check{p}_3$ & $\check{p}_6$ & $-\check{p}_5$ & $-\check{p}_8$ & $\check{p}_7$ & $\check{p}_2$ & $\check{p}_1$ & $\check{p}_4$ & $-\check{p}_3$ & $\check{p}_6$ & $-\check{p}_5$ & $-\check{p}_8$ & $\check{p}_7$\\
$-\check{p}_3$ & $-\check{p}_4$ & $-\check{p}_9$ & $\check{p}_2$ & $\check{p}_7$ & $-\check{p}_8$ & $-\check{p}_5$ & $\check{p}_6$ & $\check{p}_3$ & $-\check{p}_4$ & $\check{p}_1$ & $\check{p}_2$ & $\check{p}_7$ & $\check{p}_8$ & $-\check{p}_5$ & $-\check{p}_6$\\
$-\check{p}_4$ & $\check{p}_3$ & $-\check{p}_2$ & $-\check{p}_9$ & $\check{p}_8$ & $-\check{p}_7$ & $\check{p}_6$ & $-\check{p}_5$ & $\check{p}_4$ & $\check{p}_3$ & $-\check{p}_2$ & $\check{p}_1$ & $\check{p}_8$ & $-\check{p}_7$& $\check{p}_6$ & $-\check{p}_5$\\
$-\check{p}_5$ & $-\check{p}_6$ & $-\check{p}_7$ & $-\check{p}_8$ & $-\check{p}_9$ & $\check{p}_2$ & $\check{p}_3$ & $\check{p}_4$ & $\check{p}_5$ & $-\check{p}_6$ & $-\check{p}_7$ & $-\check{p}_8$ & $\check{p}_1$ & $\check{p}_2$ & $\check{p}_3$ & $\check{p}_4$\\
$-\check{p}_6$ & $\check{p}_5$ & $-\check{p}_8$ & $\check{p}_7$ & $-\check{p}_2$ & $-\check{p}_9$ & $-\check{p}_4$ & $\check{p}_3$ & $\check{p}_6$ & $\check{p}_5$ & $-\check{p}_8$ & $\check{p}_7$ & $-\check{p}_2$ & $\check{p}_1$ & $-\check{p}_4$ & $\check{p}_3$\\
$-\check{p}_7$ & $\check{p}_8$ & $\check{p}_5$ & $-\check{p}_6$ & $-\check{p}_3$ & $\check{p}_4$ & $-\check{p}_9$ & $-\check{p}_2$ & $\check{p}_7$ & $\check{p}_8$ & $\check{p}_5$ & $-\check{p}_6$ & $-\check{p}_3$ & $\check{p}_4$& $\check{p}_1$ & $-\check{p}_2$\\
$-\check{p}_8$ & $-\check{p}_7$ & $\check{p}_6$ & $\check{p}_5$ & $-\check{p}_4$ & $-\check{p}_3$ & $\check{p}_2$ & $-\check{p}_9$ & $\check{p}_8$ & $-\check{p}_7$ & $\check{p}_6$ & $\check{p}_5$ & $-\check{p}_4$ & $-\check{p}_3$ & $\check{p}_2$ & $\check{p}_1$\\
\end{tabular}
\right)$
\end{figure}

The transpose of this matrix is exactly what we need to decompose
the chi-square statistic. The columns are all orthogonal and the
entries still enjoy the AB-BA property, and with all non-zero
probabilities allowing for each component test to use all
multinomial cell counts.

The method above can be extended to higher Cayley-Dickson algebras
via similar calculations. Moreover, it can be applied to every
multiplication table provided by our colored matrices.

\section{Appendix}
In this appendix we show how to `color' a $2^k\times2^k$ matrix
${\bf S}_k$ with the AB-BA property (given in
Proposition~\ref{P:rectangleProperty}) in a way that produces
orthogonal columns, thus proving Theorem~\ref{T:LatinHadamard}. This
is done in two steps. First we present a systematic procedure to
color the entries so that in every set of four AB-BA corners,
exactly one is colored with a distinct color (i.e. one `+' and three
`-'s or one `-' and three `+'s). In practice, the matrix is
completely colored before checking all AB-BA corners. Therefore,
once all entries are colored, we need to check that all columns are
orthogonal. The steps detailed below yield all possible cases. For
$w=1,\cdots, k$ we will represent ${\bf S}_w$ as a block matrix
$\left[
\begin{tabular}{ll}
${\bf A_1}$ & ${\bf B_1}$\\
${\bf B_2}$ & ${\bf A_2}$
\end{tabular}
\right]$, as in Lemma~\ref{L:Latin}(c), with indices introduced for
easy referencing.

First we color all entries of the first column and first row by `+'.
The following remarks are crucial. Since the diagonal entries all
have value 1, they must obviously be colored by `-'. The diagonal
entries of ${\bf B}$ are all equal to $2^{w-1}+1$. Hence, for
$i=1,\cdots, 2^{w-1}$,
\begin{eqnarray*}
{\bf S}_{1,2^{w-1}+1}&=&{\bf S}_{i,2^{w-1}+i},\\
{\bf S}_{1,2^{w-1}+i}&=&{\bf S}_{i,2^{w-1}+1}.
\end{eqnarray*}
Since the two entries on the left are already colored with `+', and
since these four entries form AB-BA corners, the two entries on the
right must have distinct colors. Similarly, the four entries ${\bf
S}_{i,1}$, ${\bf S}_{i,2^{w-1}+i}$, ${\bf S}_{2^{w-1}+1,1}$ and
${\bf S}_{2^{w-1}+1,2^{w-1}+i}$ form a $2\times2$ submatrix whose
 entries are AB-BA corners. Since the first and third are
colored by `+' (as they belong to the first column), it is evident
that ${\bf S}_{i,2^{w-1}+i}$ and ${\bf S}_{2^{w-1}+1,2^{w-1}+i}$
must have different colors.

In a like manner, by alternately using the first row and the first
column to color a new entry, we see that the sequence of entries
\[
{\bf S}_{i,2^{w-1}+1}, {\bf S}_{i,2^{w-1}+i}, {\bf
S}_{2^{w-1}+1,2^{w-1}+i}, {\bf S}_{2^{w-1}+1,i}, {\bf
S}_{2^{w-1}+i,i}, {\bf S}_{2^{w-1}+i,2^{w-1}+1}, {\bf
S}_{i,2^{w-1}+1}.
\]

must have alternating colors. Notice that by doing this we cycle
back to the first entry, and that the sign allocation is consistent
because the number of different entries in the sequence is even
(namely 6.)

The above remarks suggest a recursive strategy to color the
$2^k\times2^k$ matrix ${\bf S}_k$ by coloring the $2^w\times2^w$
upper left submatrix progressively as $w$ increases from $2$ to $k$.

(1) Base Case: Since the first column and first row are all colored
with `+', and consequently the diagonal entries that are not on the
first row are colored with `-', the $2\times2$ upper left submatrix
is colored.

(2) Recursive Step: Suppose the $2^{w-1}\times2^{w-1}$ upper left
submatrix is colored, for some $w>2$. We choose an arbitrary color
for each entry ${\bf S}_{i,2^{w-1}+1}$, $i=2,\cdots,2^{w-1}$. By the
remark given above, each entry just colored leads to coloring a
total of six entries. Noting that two of these entries belong to the
$(2^{w-1}+1)^{st}$ row and two entries belong to the
$(2^{w-1}+1)^{st}$ column, it follows--in particular--that the
$(2^{w-1}+1)^{st}$ row and the $(2^{w-1}+1)^{st}$ column are colored
completely. Next, we are ready to color the entries of block ${\bf
B}_2$. For $i=1,\cdots,2^{w-1}$ and $j=2,\cdots,2^{w-1}$, if the
entry ${\bf S}_{2^{w-1}+i,j}$ has not been colored, we observe that
the absolute value of this entry exists in the first column of ${\bf
B}_1$, which is the $(2^{w-1}+1)^{st}$ column of ${\bf S}_w$. Hence,
for an appropriate $i^{\prime}$ in $\{1,\cdots, 2^{w-1}\}$, ${\bf
S}_{i^{\prime},2^{w-1}+1}={\bf S}_{2^{w-1}+i,j}$. It follows, by the
AB-BA property, that
\[
{\bf S}_{i^{\prime},j}={\bf S}_{2^{w-1}+i,2^{w-1}+1}.
\]
Since exactly three of the above four entries are colored, the
fourth will be colored accordingly.

At this point, using the fact that diagonal entries are `-1`, and
that the matrix ${\bf S}_w$ is symmetric, all uncolored entries of
${\bf B}_1$ are colored as follows. For $i=2,\cdots,2^{w-1}$,
$j=2,\cdots,2^{w-1}$
\[
{\bf S}_{i,2^{w-1}+j}=-{\bf S}_{2^{w-1}+j,i}.
\]
Lastly, once ${\bf A}_1$, ${\bf B}_1$,  and ${\bf B}_2$ are all
colored, coloring ${\bf A}_2$ is immediate. This concludes the
construction.

To count the number of potential colorings of a $2^k\times2^k$
matrix, observe that we start choosing random colors when $w=2$ to
$w=k$. We get $(2^1-1)+\cdots+(2^{k-1}-1)=2^k-(k+1)$. Therefore the
number of potential colorings is $2^{2^k-(k+1)}$.

Our scheme for coloring need not produce all orthogonal columns. So
further checking is needed. It turns out that the 2 colorings of
${\bf S}_2$ and the 16 colorings of ${\bf S}_3$ all yield matrices
with orthogonal columns. However, none of the 2048 colorings of
${\bf S}_4$ gives all orthogonal columns.

\end{document}